\documentclass{llncs}
\usepackage{multicol}
\usepackage{wrapfig}
\usepackage{thmtools} 
\usepackage{thm-restate}
\usepackage{url}
\usepackage{amsmath,amssymb}
\usepackage{color}
\usepackage{nicefrac}
\usepackage{graphicx}
\usepackage{multirow}
\usepackage{rotate}
\usepackage{hyperref}
\usepackage{mathtools}
\usepackage[titletoc]{appendix}

\newif\ifcommentson\commentsonfalse
\def\mywidth{.9}        
\def\mywidthRep{.8} 
\ifcommentson
\newcommand{\commentAA}[1]{\begin{center} \parbox{\mywidth\textwidth}{\textbf{\textcolor{black}{Comment A.}} \textcolor{red}{#1 }}\end{center}}
\newcommand{\commentAB}[1]{\begin{center} \parbox{\mywidth\textwidth}{\textbf{\textcolor{black}{Comment B.}} \textcolor{red}{#1} }\end{center}}
\newcommand{\commentMA}[1]{\begin{center} \parbox{\mywidth\textwidth}{\textbf{\textcolor{black}{Comment M.}} \textcolor{red}{#1} }\end{center}}
\newcommand{\replyAA}[1]{\begin{center} \parbox{\mywidthRep\textwidth}{\textbf{Reply A.} \textcolor{blue}{#1} }\end{center}}
\newcommand{\replyAB}[1]{\begin{center} \parbox{\mywidthRep\textwidth}{\textbf{Reply B.} \textcolor{blue}{#1} }\end{center}}
\newcommand{\replyMA}[1]{\begin{center} \parbox{\mywidthRep\textwidth}{\textbf{Reply M.} \textcolor{blue}{#1} }\end{center}}
\newcommand{\commentA}[1]{\marginpar{\footnotesize \color{red} {\bf A:} \textsf{\scriptsize #1}}}
\newcommand{\commentB}[1]{\marginpar{\footnotesize \color{red} {\bf B:} \textsf{\scriptsize #1}}}
\newcommand{\commentM}[1]{\marginpar{\footnotesize \color{red} {\bf M:} \textsf{\scriptsize #1}}}
\newcommand{\replyA}[1]{\marginpar{\footnotesize \color{blue} {\bf A:} \textsf{\scriptsize #1}}}
\newcommand{\replyB}[1]{\marginpar{\footnotesize \color{red} {\bf B:} \textsf{\scriptsize #1}}}
\newcommand{\replyM}[1]{\marginpar{\footnotesize \color{red} {\bf M:} \textsf{\scriptsize #1}}}
\else
\newcommand{\commentAA}[1]{}
\newcommand{\commentAB}[1]{}
\newcommand{\commentMA}[1]{}
\newcommand{\replyAA}[1]{}
\newcommand{\replyAB}[1]{}
\newcommand{\replyMA}[1]{}
\newcommand{\commentA}[1]{}
\newcommand{\commentB}[1]{}
\newcommand{\commentM}[1]{}
\newcommand{\replyA}[1]{}
\newcommand{\replyB}[1]{}
\newcommand{\replyM}[1]{}
\fi

\newcommand{\calw}{\mathcal{W}}
\newcommand{\calx}{\mathcal{X}}
\newcommand{\caly}{\mathcal{Y}}
\newcommand{\calz}{\mathcal{Z}}

\newcommand{\calc}{\mathcal{C}}

\newcommand{\cupdot}{\mathbin{\mathaccent\cdot\sqcup}}
\newcommand{\hchoice}[1]{\;{{}_{\mathit{#1}}{\oplus}}\;} 
\newcommand{\hchoiceop}[1]{{}_{\mathit{#1}}{\oplus}} 
\newcommand{\vchoice}[1]{\;{{}_{\mathit{#1}}{\cupdot}}\;} 
\newcommand{\vchoiceop}[1]{{}_{\mathit{#1}}{\cupdot}} 

\newcommand{\eqbij}{\;\overset{\circ}{=}\;}

\newcommand{\reals}{\mathbb{R}}
\newcommand{\dist}{\mathbb{D}}

\newcommand{\qm}[1]{``#1''}

\newcommand{\hyperc}[2]{[#1 \, \rangle \, #2]} 

\newcommand{\nullchannel}{\overline{0}} 
\newcommand{\transparentchannel}{\overline{I}} 

\renewcommand{\equiv}{\approx}
\newcommand{\refines}{\sqsubseteq_{\circ}}

\newcommand{\review}[1]{#1}

\newcommand{\version}[2]{#2}

\begin{document}
\title{An Algebraic Approach for Reasoning About Information Flow}

\author{Arthur Am\'{e}rico\inst{1}
\and M{\'a}rio S. Alvim\inst{1}
\and Annabelle McIver\inst{2}}
\authorrunning{Am\'{e}rico et al.}
\institute{Universidade Federal de Minas Gerais, Belo Horizonte, Brazil
\and Macquarie University, Sydney, Australia}

\maketitle
\begin{abstract}
This paper concerns the analysis of information leaks in security systems.
We address the problem of specifying and analyzing large systems in the (standard)
channel model used in quantitative information flow (QIF). 
We propose several operators which match typical interactions between system 
components. 
We explore their algebraic properties with respect to the security-preserving
refinement relation defined by Alvim et al. and McIver et al.~\cite{Alvim:12:CSF,McIver:14:POST}.

We show how the algebra can be used to simplify large system specifications in order to
facilitate the computation of information leakage bounds. 
We demonstrate our results on the specification and analysis of the Crowds Protocol. 
Finally, we use the algebra to justify a new algorithm to compute leakage bounds for 
this protocol.
 


\end{abstract}

\section{Introduction}
\label{sec:introduction}
\commentM{Remember to remove authors' name in case the submission is supposed to be double-blind.}
\commentM{I added pagestyle{plain} to the header so pages are numbered.
You can comment it out for submission.}
Protecting sensitive information from unintended disclosure is a crucial goal for 
information security.
There are, however, many situations in which information leakage is unavoidable. 
An example \review{is a typical password checker}, which must always reveal some information
about the secret password---namely whether or not it matches the input provided by the 
user when trying to log in.
Another example \review{concerns} election tallies, which reveal information about individual 
votes by ruling out several configurations of votes (e.g., in the extreme case of an 
unanimous election, the tally reveals every vote).
The field of \emph{Quantitative Information Flow} (QIF) is concerned with quantifying 
the amount of sensitive information computational systems leak, and it has been extremely 
active in the past decade \cite{Clark:05:JLC,Kopf:07:CCS,Chatzikokolakis:08:JCS,Smith:09:FOSSACS,McIver:10:ICALP,Boreale:15:LMCS,axioms}.

In the QIF framework, systems are described as receiving \emph{secret inputs} from a set of 
values $\calx$, and producing  \emph{public}, or \emph{observable}, \emph{outputs} from 
a set $\caly$. 
Typical secret inputs are a user's identity, password, or current location, whereas public outputs are anything an adversary can observe about the behavior of the system, such as messages written
on the screen, execution time, or power consumption. 
A system is, then, modeled as an \emph{(information-theoretic) channel}, which is a function 
mapping each possible pair $x \in \calx$, $y \in \caly$ to the conditional probability 
$p(y \mid x)$ of the system producing output $y$ when receiving input $x$. 
Channels abstract technicalities of the system, while retaining the essentials 
that influence information leakage: the relation between secret input and public output values.

The QIF framework provides a robust theory for deriving security properties from a system's 
representation as a channel. 
However, obtaining an appropriate channel to model a system is often a non-trivial task.
Moreover, some channels turn out to be so large as to render most security analyses unfeasible 
in practice.

In this paper we provide an algebra for describing (larger, more complex) channels as 
compositions of other (smaller, simpler) channels. 
For that, we define a set of operators, each corresponding to a different way in which components 
can interact in a system---namely, \emph{parallel} composition, 
\emph{visible choice} composition, and \emph{hidden choice} composition.
We prove a series of algebraic properties of these operators, 
and use such properties to simplify system specifications so that bounds 
on the information leakage of a compound system can be inferred from the
information leakage of its components.
In this way, we allow for leakage analyses of systems which would be 
intractable with traditional QIF techniques.

This compositional approach seems particularly natural for modeling security protocols,
which often involve interactions among various entities.
Consider, for instance, the well-known \emph{Dining Cryptographers} 
anonymity protocol~\cite{dining}. 
A group of $n$ cryptographers has been invited for dinner by the NSA 
(American National Security Agency), who will either pay the bill, or secretly ask 
one of the cryptographers to be the payer.
The cryptographers want to determine whether one among them is the payer, 
but without revealing which one. 
For that, they execute the following protocol.
In a first phase all participants form a circle, and each tosses a coin and shares the result only with the cryptographer on his right.
In a second phase, each cryptographer computes the exclusive-or of the two coins 
tosses he observed (interpreting \emph{heads} as $0$ and \emph{tails} as $1$), 
and publicly announces the result.
The only exception is the paying cryptographer (if any), who announces the 
negation of his exclusive-or.
In a third phase, the cryptographers compute the exclusive-or of all announcements. One of them is the payer if, and only if, the 
result is $1$.
It has been shown that, if all coins are fair, no information is leaked about who the paying cryptographer is~\cite{dining}.

\begin{wrapfigure}{r}{0.5\linewidth}
\centering
\vspace{-7mm}
\includegraphics[width=0.45\textwidth]{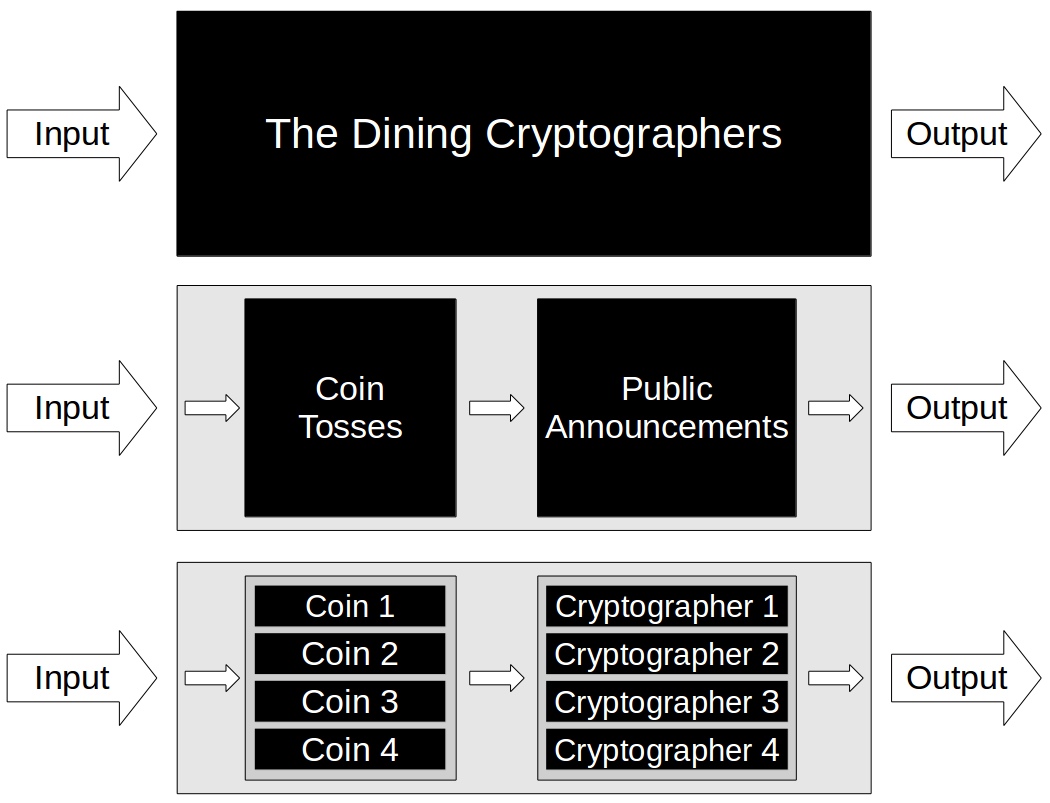}
\label{fig:scheme}
\vspace{-2mm}
\caption{Schematic representation of the Dining Cryptographers protocol as: 
(i) a monolithic channel (top); 
(ii) a composition of two channels (middle); and
(ii) a composition of eight channels (bottom).}
\vspace{-6mm}
\end{wrapfigure}
Despite the Dining Cryptographers relative simplicity, 
deriving its channel can be a challenging task.
Since each of the $n$ cryptographers can announce either $0$ or $1$, 
the size of the output set $\caly$, and, consequently, of the channel, 
increases exponentially with the number of cryptographers.
The problem is worsened by the fact that computing the probabilities 
constituting the channel's entries is not trivial.
The algebra we introduce in this paper allows for an intuitive and compositional way
of building a protocol's channel \review{from} each of its components.
To illustrate the concept, Figure~\ref{fig:scheme} depicts three alternative representations, using 
channels, for the Dining Cryptographers with $4$ cryptographers and $4$ coins.
In all models, the input is the identity of the payer (one
of the cryptographers or the NSA), and the output are the public announcements of all cryptographers.
The top model uses a single (enormous) channel to represent the protocol; 
the middle one models the protocol as the interaction between
two smaller components (the coins and the party of cryptographers);
the bottom one uses interactions between even smaller channels (one for each coin and each cryptographer).


The main contributions of this paper are the following.

\begin{itemize}
\item \review{We formalize several common operators for channel composition 
used in the literature, each matching a typical interaction between system components. 
We prove several relevant algebraic and information-leakage properties of these operators.} 

\item \review{We show that the substitution of 
components in a system may be subject to unexpected, and 
perhaps counter-intuitive, results.
In particular, we show that overall leakage may increase even 
when the new component is more secure than the one it is 
replacing (e.g., Theorems \ref{theorem:parord} and \ref{theorem:ressechchoicei}).}

\item We show how the proposed algebra can be used to simplify large 
system specifications in order to facilitate the computation of 
information leakage bounds, given in terms of the $g$-leakage framework~\cite{Alvim:12:CSF,Alvim:14:CSF,Alvim:16:CSF,McIver:14:POST}.

\item We demonstrate our results on the specification and analysis of the 
Crowds Protocol~\cite{crowds}. 
We use the proposed algebra to justify a new algorithm to compute leakage bounds for this protocol.
\end{itemize}

Detailed proofs of all of our technical results can be found \review{\version{in an accompanying technical report\cite{arxiv}}{in Appendix \ref{sec:proofs}}}.
\paragraph{Plan of the paper.} 
The remainder of this paper is organized as follows. 
In Section~\ref{sec:preliminaries} we review fundamental concepts from QIF.
In Section~\ref{sec:oper} we introduce our channel operators, 
and in Section~\ref{sec:algebraic} we provide their algebraic properties.
In Section~\ref{sec:leakage} we present our main results, concerning 
information leakage in channel composition. 
In Section~\ref{sec:casestudy} we develop a detailed case study of the 
Crowds protocol.
Finally, in Section~\ref{sec:related} we discuss related work, and in
Section~\ref{sec:conc} we conclude. 

\section{Preliminaries}
\label{sec:preliminaries} 

In this section we review some fundamentals from quantitative 
information flow.

\paragraph{Secrets, gain functions and vulnerability.}
A \emph{secret} is some piece of sensitive information that 
one wants to protect from disclosure.
Such sensitive information may concern, for instance, a \review{user's} 
password, identity, personal data, or current location.
We represent by $\calx$ the set of possible \emph{secret values} 
the secret may take. 

The \emph{adversary} is assumed to have, 
before observing the system's behaviour, 
some \emph{a priori} partial knowledge about the secret value.
This knowledge is modeled as a probability distribution 
$\pi \in \dist \calx$, where $\dist \calx$ denotes the set of all 
probability distributions on $\calx$. 
We call $\pi$ a \emph{prior distribution}, or simply a \emph{prior}.

To quantify how \emph{vulnerable} a secret is---i.e., how prone it is
to exploitation by the adversary\review{---} we employ a function that
maps probability distributions to the real numbers (or, more in general,
to any ordered set).
Many functions have been used in the literature, such as \emph{Shannon entropy} \cite{Shannon:48:Bell}, guessing-entropy~\cite{Massey:94:IT}, Bayes vulnerability~\cite{Braun:09:MFPS}, and \emph{R\'{e}nyi min-entropy}~\cite{Smith:09:FOSSACS}.
Recently, the \emph{$g$-leakage} \cite{Alvim:12:CSF} framework was proposed, 
and it proved to be very successful in capturing a variety of different scenarios, including those in which the
adversary benefits from guessing part of secret, guessing
a secret approximately, guessing the secret within 
a number of tries, or gets punished for guessing wrongly.
In particular, the framework has been shown to be able to capture all
functions mentioned above~\cite{Alvim:16:CSF}.
In this framework, a \review{finite} set $\calw$ of \emph{actions} is available to the adversary,
and a \emph{gain-function} $g{:}\calw{\times}\calx{\rightarrow}[0,1]$  
is used to describe the benefit $g(w,x)$ an adversary obtains when he performs
action $w{\in}\calw$, and the secret value is $x{\in}\calx$.
Given \review{an appropriate} gain-function $g$, the secret's \emph{(prior) $g$-vulnerability}  
is defined as the expected value of the adversary's gain if he chooses a 
best possible action,
\begin{equation*}
V_g [\pi]= \max\limits_{w \in \calw} \sum_{x \in \calx} \pi(x) g(w,x),
\end{equation*}
and the greater its value, the more vulnerable, or insecure, the secret is.

\paragraph{Channels and posterior vulnerabilities}
In the QIF framework, a system is usually modeled as an 
\emph{(information theoretic) channel} taking a \emph{secret input} $x{\in} \calx$, 
and producing a \emph{public}, or \emph{observable}, \emph{output} $y{\in} \caly$.
Each element of $\caly$ represents a behaviour from the system 
that can be discerned by the adversary.
Formally, a channel is a function $C{:}\calx{\times}\caly{\rightarrow}\reals$
such that $C(x,y)$ is the conditional probability \review{$p(y {\mid} x)$} of the system 
producing output $y{\in}\caly$ when input is $x{\in}\calx$.

A channel $C$ together with a prior $\pi$ induce a joint probability distribution
$p$ on the set $\calx{\times}\caly$, given by $p(x,y)=\pi(x)C(x,y)$.
From this joint distribution we can derive, for every $x{\in}\calx$ and $y{\in}\caly$,
\review{the marginal probabilities $p(x)=\sum_{y}p(x,y)$ and $p(y)=\sum_{x}p(x,y)$, and the conditional probabilities 
$p(x {\mid} y)=\nicefrac{p(x,y)}{p(y)}$ and $p(y {\mid} x )=\nicefrac{p(x,y)}{p(x)}$}. 
Note that $p(x)=\pi(x)$ and, if $p(x)\neq 0$, $p(y {\mid} x)=C(x,y)$. 

By observing the output produced by the system, the adversary can update
his knowledge about the secret value.
More specifically, if the system outputs $y{\in}\caly$, an adversary can 
update the prior $\pi$ to a revised \emph{posterior distribution}
$p_{X{\mid}y} \in \dist \calx$ on $\calx$ given $y$, defined 
for all $x{\in}\calx$ and $y{\in}\caly$ as $p_{X{\mid}y}(x)=p(x{\mid}y)$.

\begin{example}
\label{exa:hypers}
Let $\calx=\{x_{1},x_{2},x_{3}\}$ and $\caly=\{y_{1},y_{2},y_{3},y_{4}\}$
be input and output sets.
Let $\pi=(\nicefrac{1}{2},\nicefrac{1}{3}, \nicefrac{1}{6})$ be a prior,
and $C$ be the channel below. 
The combination of $\pi$ and $C$ yield a joint probability $p$, 
according to the tables below.
\begin{align*}
\begin{array}{|c|cccc|}
\hline
C & \,\,\,y_{1}\,\,\, & \,\,\,y_{2}\,\,\, & \,\,\,y_{3}\,\,\, & \,\,\,y_{4}\,\,\,
\\ \hline
x_{1} & \nicefrac{1}{6} & \nicefrac{2}{3} & \nicefrac{1}{6} & 0 \\
x_{2} & \nicefrac{1}{2} & \nicefrac{1}{4} & \nicefrac{1}{4} & 0 \\
x_{3} & \nicefrac{1}{2} & \nicefrac{1}{3} & 0 & \nicefrac{1}{6} \\
\hline
\end{array}
\quad
\stackrel{\pi}{\longrightarrow}
\quad
\begin{array}{|c|cccc|}
\hline
p & \,\,\,y_{1}\,\,\, & \,\,\,y_{2}\,\,\, & \,\,\,y_{3}\,\,\, & \,\,\,y_{4}\,\,\, 
\\ \hline
x_{1} & \nicefrac{1}{12} & \nicefrac{1}{3} & \nicefrac{1}{12} & 0 \\
x_{2} & \nicefrac{1}{6} & \nicefrac{1}{12} & \nicefrac{1}{12} & 0 \\
x_{3} & \nicefrac{1}{12} & \nicefrac{1}{18} & 0 & \nicefrac{1}{36} \\
\hline
\end{array}
\end{align*}
By summing the columns of the second table, we obtain the marginal probabilities 
$p(y_1){=}\nicefrac{1}{3}$, $p(y_2){=}\nicefrac{17}{36}$, $p(y_3){=}\nicefrac{1}{6}$ and $p(y_4){=}\nicefrac{1}{36}$.
These marginal probabilities yield the posterior distributions
$p_{X \mid y_{1}}{=}(\nicefrac{1}{4},\nicefrac{1}{2},\nicefrac{1}{4})$,
$p_{X \mid y_{2}}{=}(\nicefrac{12}{17},\nicefrac{3}{17},\nicefrac{2}{17})$,
$p_{X \mid y_{3}}{=}(\nicefrac{1}{2},\nicefrac{1}{2},0)$, and
$p_{X \mid y_{4}}{=}(0,0,1)$.
\qed
\end{example}

The \emph{posterior $g$-vulnerability} of a prior $\pi$ and a channel $C$
is defined as the expected value of the secret's $g$-vulnerability after the 
execution of the system:
\begin{equation*}
V_g\hyperc{\pi}{C}= \sum_{y \in \caly} \max\limits_{w \in \calw} \sum_{x \in \calx} C(x,y)\pi(x)g(x,w).
\end{equation*}

The \emph{information leakage} of a prior and a channel is a measure
of the increase in secret vulnerability caused by the observation
of the system's output.
Leakage is, thus, defined as a comparison between the secret's prior and 
posterior vulnerabilities.
Formally, for a gain-function $g$, and given prior $\pi$ and channel $C$, the \emph{multiplicative} and the \emph{additive} versions of $g$-leakage are defined, respectively, as
\begin{equation*}
\mathcal{L}_{g}\hyperc{\pi}{C}= \nicefrac{V_g\hyperc{\pi}{C}}{V_g[\pi]}, \qquad \text{and} \qquad
\mathcal{L}_{g}^+\hyperc{\pi}{C}=V_g\hyperc{\pi}{C}-V_g[\pi].
\end{equation*}

Since prior vulnerability does not depend on the channel, we have that
\begin{align*}
\mathcal{L}_{g}\hyperc{\pi}{C_1}\geq \mathcal{L}_{g}\hyperc{\pi}{C_2} 
\,\,\Leftrightarrow\,\,
\mathcal{L}_{g}^+\hyperc{\pi}{C_1}\geq \mathcal{L}_{g}^+\hyperc{\pi}{C_2} \,\,\Leftrightarrow\,\,
V_g\hyperc{\pi}{C_1}\geq V_g\hyperc{\pi}{C_2},
\end{align*}
and, hence, the posterior vulnerability of a channel is greater
than that of another if, and only if, both multiplicative and
additive leakage also are.

\paragraph{Channel Ordering and the Coriaceous Theorem.}
We now define a common composition of channels, called \emph{cascading}.
This operation can be interpreted as the result of a channel post-processing 
the output of another channel. 
Formally, given two channels $C{:}\calx{\times}\caly{\rightarrow}\reals$ and $D{:}\caly{\times}\calz{\rightarrow}\reals$,
their cascading is defined as
\begin{equation*}
(CD)(x,z)= \sum_{y \in \caly} C(x,y)D(y,z),
\end{equation*}
for all $x{\in}\calx$ and $z{\in}\calz$. 
If we represent channels as tables, as we did in Example~\ref{exa:hypers}, 
the cascading operation corresponds to a simple matrix multiplication.

An important question in QIF is to decide whether a channel $C_{2}$ is always 
\emph{at least as secure as} a channel $C_{1}$, meaning that $C_{2}$ never leaks 
more information than $C_{1}$, for whatever choice of gain function $g$ and of 
prior $\pi$.
Let us write $C_1 \refines C_2$ (read as $C_{2}$ \emph{refines} $C_{1}$) 
to denote that there \review{exists} a channel $D$ such that $C_1D=C_2$.
We write $C_1 \equiv C_2$, and say that $C_1$ is \emph{equivalent} to $C_2$, 
when both $C_1 \refines C_2$ and $C_2 \refines C_1$ hold.
The \emph{Coriaceous Theorem} \cite{Alvim:12:CSF,McIver:14:POST} states that, 
$C_1 \refines C_2$ if, and only if, 
$V_g \hyperc{\pi}{C_1}\geq V_g \hyperc{\pi}{C_2}$ for all $\pi$, $g$. 
This result reduces the comparison of channel security to a 
simple algebraic test.

The \emph{refinement relation} $\refines$ is a preorder on the set of all channels having \review{the} same input set. 
This preorder can be made into a partial order by using \emph{abstract channels}
\cite{McIver:14:POST}, an equivalence relation that equates all channels presenting
same leakage for all priors and gain functions.
This partial order coincides with how much information channels leak,
being the least secure channel (i.e., the \qm{most leaky} one) at its bottom,
and the most secure (i.e., the \qm{least leaky}) at its top.


\section{Operators on channel composition}
\label{sec:oper}
\review{We shall say that two channels are \emph{compatible} if they have the same input set.
Given a set $\calx$, we denote by $\calc_\calx$ the set of all 
channels that have $\calx$ as input set.
Two compatible channels with same output set are said to be of the \emph{same type}.}

\review{In this section we introduce several \emph{binary operators}---i.e., functions of type 
$(\calc_\calx{\times}\calc_\calx){\rightarrow}\calc_\calx$---matching typical interactions between system components, and prove relevant algebraic properties of these operators.
We refer to the result of an operator as a \emph{compound system},
and we refer to its arguments as \emph{components}.}

\subsection{The parallel composition operator $\parallel$}

The \emph{parallel composition operator} $\parallel$
models the composition of two \review{independent} channels in which the
same input is fed to both of them, and their outputs
are then observed. \review{By \emph{independent}, we mean that the output of one channel does not interfere with that of the other.
This assumption, while not universal, captures a 
great variety of real-world scenarios, 
and is, hence, of practical interest.}

For example, \emph{side-channel attacks} occur when
the adversary combines his observation of the 
system's output with some alternative way of
inferring information about the secret 
(e.g., by observing physical properties of the system
execution,  such as time elapsed~\cite{Kosher:96:CRYPTO,Brumley:03:USENIX} or change in magnetic fields~\cite{Nohl:08:CSS}). 
In such attacks, the channel used by the adversary
to infer information about the secret can be modeled as the composition of a channel representing the program's intended behaviour in parallel with a channel
modeling the relation between the secret and the 
physical properties of the hardware. 


\begin{definition}[Parallel composition operator $\parallel$]
Given compatible channels 
$C_1{:}\calx{\times}\caly_1{\rightarrow}\reals$ and 
$C_2{:}\calx{\times}\caly_2{\rightarrow}\reals$, their 
\emph{parallel composition} 
$C_1 \parallel C_2 : \calx{\times}(\caly_1{\times}\caly_2){\rightarrow}\reals$ is defined as, for all $x{\in}\calx$, $y_1{\in}\caly_1$, and $y_2{\in}\caly_2$,
$$(C_1 \parallel C_2)(x,(y_1,y_2))=C_1(x,y_1)  C_2(x,y_2).$$ 
\end{definition}

\review{Notice that this definition comes from the independence property, as we have $C_1(x,y_1)  C_2(x,y_2)=p(y_1 {\mid} x)p(y_2 {\mid} x)=p(y_1,y_2 {\mid} x)$.}
\subsection{The visible choice operator $\vchoiceop{p}$}

The \emph{visible choice operator} $\vchoiceop{p}$ models a
scenario in which the system has a choice among two different components to process the secret it was fed as input.
With probability $p$, the system feeds the secret
to the first component, and, with probability $1{-}p$, it
feeds the secret to the second component.
In the end, the system reveals the output produced, together with the identification of which component was used (whence, the name \qm{visible choice}).

As an example, consider an adversary trying to gain
information about a secret processed by a website.
The adversary knows that the website has two servers,
one of which will be assigned to answer the request
according to a known probability distribution.
Suppose, furthermore, that the adversary can
identify which server was used by measuring its 
response time to the request.
This adversary's view of the system can be modeled as
the visible choice between the two servers, since, although
the adversary does not know in advance which server will be used, he learns it when he gets the
output from the system.


Before formalizing this operator, we need to define the \emph{disjoint union} of sets.  Given any sets $\mathcal{A}$ and $\mathcal{B}$, their disjoint union is
$\mathcal{A} \sqcup \mathcal{B}= (\mathcal{A}\times \{1\})\cup (\mathcal{B}\times \{2\})$.

\begin{definition}[Visible choice operator $\vchoiceop{p}$]
Given compatible channels 
$C_1:\calx{\times}\caly_1{\rightarrow}\reals$ and 
$C_2{:}\calx{\times}\caly_2{\rightarrow}\reals$, their
\emph{visible choice} is the channel 
$C_1 \vchoice{p} C_2:\calx{\times}(\caly_1 \sqcup \caly_2){\rightarrow}\reals$ defined as, for all $x{\in}\calx$ and $(y,i) \in \caly_1 \sqcup \caly_2$,
$$(C_1 \vchoice{p} C_2)(x,(y,i))=
\begin{cases}
p  C_1(x,y), & \mbox{if } i = 1, \\
(1{-}p)  C_2 (x,y), & \mbox{if } i=2.
\end{cases}$$ 
\end{definition}

\subsection{The hidden choice operator $\hchoiceop{p}$}

Similarly to the visible choice case, the 
\emph{hidden choice operator} $\hchoiceop{p}$ models a
scenario in which the system has a choice of feeding
its secret input to one component (with probability $p$),
or to another component (with probability $1{-}p$).
In the end, the system reveals the output produced,
but, unlike the visible choice case, the component 
which was used is not revealed.
Hence, when the same observations are randomized
between the two channels, the adversary cannot
identify which channel produced the observation
(whence, the name \qm{hidden choice}).

As an example, consider statistical surveys that ask some sensitive 
yes/no question, such as whether the respondent has made use of any illegal substances.
To encourage individuals to participate on the survey,
it is necessary to control leakage of their sensitive information, while preserving the accuracy of statistical
information in the ensemble of their answers.
A common protocol to achieve this goal works as follows~\cite{Warner:65:JASA}.
Each respondent throws a coin, without letting the questioner know the corresponding result.
If the result is heads, the respondent answers the
question honestly, and if the result is tails, he
gives a random response (obtained, for example, 
according to the result of a second coin toss).
If the coins are fair, this protocol can be modeled as 
the hidden choice $T{\hchoice{\nicefrac{1}{2}}}C$ between 
a channel $T$ representing an honest response
(revealing the secret completely), 
and a channel $C$ representing a random response 
(revealing nothing about the secret).
The protocol is, hence, a channel that masks
the result of $T$.

\begin{definition}[Hidden choice operator $\hchoiceop{p}$]
Given compatible channels $C_1:\calx{\times}\caly_1{\rightarrow}\reals$ and $C_2{:}\calx{\times}\caly_2{\rightarrow}\reals$, their
\emph{hidden choice} is the channel 
$C_1 \hchoice{p} C_2:\calx{\times}(\caly_1 \cup \caly_2 ){\rightarrow}\reals$ defined as, for all $x{\in}\calx$ and $y{\in}\caly_1 \cup \caly_2$,
$$(C_1 \hchoice{p} C_2)(x,y)=
\begin{cases}
p   C_1(x,y)+(1{-}p)   C_2 (x,y), & \mbox{if } y \in \caly_1 \cap \caly_2, \\
p   C_1(x,y), & \mbox{if } y \in \caly_1 \setminus \caly_2,\\
(1{-}p)   C_2 (x,y), & \mbox{if } y \in \caly_2 \setminus \caly_1.
\end{cases}$$ 
\end{definition}

Note that when the output sets of $C_1$ and $C_2$ \review{are} disjoint the adversary can \review{always  identify the channel
used}, and we have
$C_1{\vchoice{p}}C_2{\equiv}C_1{\hchoice{p}}C_2$.

\subsection{A compositional description of the Dining Cryptographers}

We now revisit the Dining Cryptographers protocol example from Section~\ref{sec:introduction}, showing how it can be modeled using our composition operators.

We consider that there are $4$ cryptographers and $4$ coins,
and denote the protocol's channel by $\mathit{Dining}$.
The channel's input set is $\calx = \{c_1,c_2,c_3,c_4,n\}$, 
in which $c_{i}$ represents that cryptographer $i$ is the payer, 
and $n$ represents that the NSA is the payer.
The channel's output set is $\caly=\{0,1\}^4$, i.e.,
all $4$-tuples representing possible announcements by all
cryptographers, in order.

Following the scheme in Figure~\ref{fig:scheme} (middle), we begin by modeling the protocol as the interaction between two channels, $\mathit{Coins}$
and $\mathit{Announcements}$, representing, respectively, the coin tosses and the cryptographers' public announcements.
Since in the protocol first the coins are tossed, and only then 
the corresponding results are passed on to the party of cryptographers, $\mathit{Dining}$ can be described as the cascading of these two channels:
$$\mathit{Dining} = (\mathit{Coins}) (\mathit{Announcements}).$$

\begin{wraptable}{r}{0.3\linewidth}
\centering
\vspace{-7mm}
\begin{tabular}{|c|cc|}
\hline
\,$\mathit{Coin}_i$\,&\, Tails\, &\, Heads \, \\
\hline
$c_1$&$p_i$&$1{-}p_i$\\
$c_2$&$p_i$&$1{-}p_i$\\
$c_3$&$p_i$&$1{-}p_i$\\
$c_4$&$p_i$&$1{-}p_i$\\
$n$  &$p_i$&$1{-}p_i$\\
\hline
\end{tabular}
\vspace{-2mm}
\caption{Channel representing toss of coin $\mathit{Coin}_{i}$.}
\label{tab:coins}
\vspace{-6mm}
\end{wraptable}
To specify channel $\mathit{Coins}$, we use the parallel composition of channels
$\mathit{Coin}_1$, $\mathit{Coin}_2$, $\mathit{Coin}_3$ and $\mathit{Coin}_4$, each representing one coin toss.
Letting $p_i$ denote the probability of 
coin $i$ landing on tails, these channels 
are defined as on Table~\ref{tab:coins}.

Besides the result of the tosses, $\mathit{Coins}$ also needs 
to pass on to $\mathit{Announcements}$ the identity of the payer.
We then introduce a fifth channel,
$I{:}\calx{\times}\calx{\rightarrow}\reals$, that simply outputs 
the secret, i.e., $I(x_1,x_2)=1$ if $x_1=x_2$, and $0$ otherwise.
Hence, a complete definition of channel $\mathit{Coins}$ is
$$\mathit{Coins}= \mathit{Coin}_1 \parallel \mathit{Coin}_2 \parallel \mathit{Coin}_3 \parallel \mathit{Coin}_4 \parallel I.$$

As we will show in Section~\ref{sec:algebraic}, parallel
composition is associative, allowing us to omit parentheses 
in the equation above.
 
We now specify the channel $\mathit{Announcements}$, which
takes as input a $5$-tuple with five terms whose
first four elements are the results of the coin tosses, 
and the fifth is the identity of the payer.
For that end, we describe each cryptographer
as a channel with this $5$-tuple as input, and with the set of possible announcements $\{0,1\}$ as output set.
$\mathit{Crypto}_1$ below describes the first cryptographer.
$$
\mathit{Crypto}_1 (t_1,t_2,t_3,t_4, x)=
\begin{cases}
1 \mbox{, if } t_4=t_1 \mbox{ and } x=c_1 \mbox{, or }t_4\neq t_1 \mbox{ and } x \neq c_1 \\
0 \mbox{, otherwise}
\end{cases}
$$

Channels $\mathit{Crypto}_2$, $\mathit{Crypt}o_3$ and $\mathit{Crypto}_4$ describing the remaining cryptographers are
defined analogously.
Channel $\mathit{Announcements}$ is, hence, defined as
$$
\mathit{Announcements} = \mathit{Crypto}_1 \parallel \mathit{Crypto}_2 \parallel \mathit{Crypto}_3  \parallel \mathit{Crypto}_4.$$

Note that our operators allow for an intuitive and succinct representation of the channel $\mathit{Dining}$ modeling the Dining Cryptographers protocol, even when the number of cryptographers and coins is large.
Moreover, the channel is easy to compute: 
we need only to first calculate the parallel compositions
within channels $\mathit{Crypto}$ and $\mathit{Announcements}$, 
and then multiply these channels' matrices.

\section{Algebraic properties of channel operators}
\label{sec:algebraic}
In this section we prove a series of relevant algebraic properties of our channel operators.
These properties are the key for building channels
in a compositional way, and, more importantly, for
deriving information flow properties of a compound
system in \review{terms} of those of its components.

We begin by defining a notion of equivalence stricter than $\equiv$, which equates any two 
channels that are identical modulo a 
permutation of their columns. 
 
\begin{definition}[Channel equality]
Let $C_1 : \calx{\times}\caly_1{\rightarrow}\reals$ and 
$C_2 : \calx{\times}\caly_2{\rightarrow}\reals$ be compatible channels. 
We say that $C_1$ and $C_2$ are \emph{equal \review{up to a permutation} }, and write 
$C_1{\review{\eqbij}}C_2$, \review{ if there is a bijection $\psi{:}\caly_1{\rightarrow}\caly_2$ 
such that 
$C_1(x,y){=}C_2(x, \psi(y))$
for all $x{\in}\calx$, $y{\in}\caly_1$.}
\end{definition}

Note that, if  
$C_1{\review{\eqbij}}C_2$, then $C_1{\equiv}C_2$.~\footnote{
\review{A complete, formal definition of such bijections can be found} \review{\version{in an accompanying technical report\cite{arxiv}}{in Appendix \ref{sec:proofs}}}.}

In remaining of this section, let $C_1:\calx{\times}\caly_1 {\rightarrow}\reals$, $C_2:\calx{\times}\caly_2{\rightarrow} \reals$ and $C_3:\calx{\times}\caly_3{\rightarrow}\reals$ be compatible channels, and $p,q \in [0,1]$ be probability values. 

\subsection{Properties regarding channel operators}


We first establish our operators' associativity and commutativity properties.



\begin{restatable}[Commutative Properties]{proposition}{rescommutative}
\label{prop:commutative}
$$
C_1{\parallel}C_2\review{{\eqbij}} C_2{\parallel}C_1, \quad
C_1{\vchoice{p}}C_2 \review{{\eqbij}} C_2{\vchoice{(1-p)}}C_1, \,\, \text{and} \,\,\,\,
C_1{\hchoice{p}}C_2 = C_2{\hchoice{(1-p)}}C_1.
$$
\end{restatable}

\begin{restatable}[Associative Properties]{proposition}{resassociative}
\label{prop:associative}
$$
\begin{array}{cc}
(C_1 \parallel C_2) \parallel C_3 \review{{\eqbij}} C_1 \parallel (C_2 \parallel C_3), \,\,\,\,\, & \,\,\,\,\,
(C_1 \vchoice{p} C_2) \vchoice{q} C_3 \review{{\eqbij}} C_1 \vchoice{p'} (C_2 \vchoice{q'} C_3),
\end{array}
$$
\vspace{-5mm}
$$
\begin{array}{cc}
\text{and } & (C_1{\hchoice{p}}C_2){\hchoice{q}}C_3 = C_1{\hchoice{p'}}
(C_2{\hchoice{q'}}C_3),
\end{array}
\text{ s.t. $p'{=}pq$ and $q'{=}\nicefrac{(q-pq)}{(1-pq)}$.}
$$
\end{restatable}


We now turn our attention to two kinds of channels
that will be recurrent building blocks for more complex channels.
A \emph{null channel} is any channel $\nullchannel: \calx{\times}\caly{\rightarrow}\reals$ such that, 
for every prior $\pi$ and gain-\review{function} $g$, $V_g\hyperc{\pi}{\nullchannel}=V_g[\pi]$.
That is, a null channel never leaks any information.
A channel $\nullchannel$ is null if, and only if, 
$\nullchannel(x,y)=\nullchannel(x',y)$ for all $y{\in}\caly$ and $x, x'{\in}\calx$.
On the other hand, a \emph{transparent channel} is any
channel $\transparentchannel: \calx{\times}\caly{\rightarrow}\reals$ that leaks at least as much information as any other compatible channel, for every prior 
and gain-function.
A channel $\transparentchannel$ is transparent if, and only if, for each $y{\in}\caly$,  there is at most one $x{\in}\calx$ such that $\transparentchannel(x,y){>}0$. 
The following properties hold for any null channel $\nullchannel$ and transparent channel $\transparentchannel$ compatible with $C_1$, $C_2$ and $C_3$.


\begin{restatable}[Null and Transparent Channel Properties]{proposition}{resnull} \label{prop:null}
$$
\begin{array}{rlll}
\text{null channel:} \quad &
(C_1 \parallel \nullchannel) \equiv C_1, \quad &
C_1 \refines (C_1 \vchoice{p} \nullchannel), \quad & 
C_1 \refines (C_1 \hchoice{p} \nullchannel). \\[2mm]
\text{transparent channel:} \quad &
(C_1 \parallel \transparentchannel) \equiv \transparentchannel, & (C_1 \vchoice{p} \transparentchannel) \refines C_1. &
\end{array}
$$
\end{restatable}

Note that, in general, $(C_1 \hchoice{p} \transparentchannel) \not\refines C_1$.
To see why, consider the two transparent channels $\transparentchannel_1$ and  $\transparentchannel_2$, with 
both input and output sets equal $\{1,2\}$, given by
$\transparentchannel_1 (x,x')=1$ if $x{=}x'$, and $0$ otherwise,
and $\transparentchannel_2 (x,x')=0$ if $x{=}x'$, and $1$ otherwise
Then, $\transparentchannel_1 \hchoice{p} \transparentchannel_2$ is a null channel,
and the property does not hold for $C_1=\transparentchannel_1$, $\transparentchannel=\transparentchannel_2$.


We now consider idemptotency.

\begin{restatable}[Idempotency]{proposition}{residempotency}\label{prop:idem}
$$
C_1 \parallel C_1 \refines C_1, \qquad
C_1 \vchoice{p} C_1 \equiv C_1, \qquad \text{and} \qquad
C_1 \hchoice{p} C_1 = C_1.
$$
\end{restatable}

Note that  $C_1{\parallel}C_1{\equiv}C_1$ holds
only when $C_1$ is deterministic 
or equivalent to a deterministic channel.

Finally, we consider distributive properties.
In particular, we explore interesting properties when 
\review{an} operator is ``distributed'' over itself.

\begin{restatable}[Distribution over \review{the} same operator]{proposition}{resdistsame}
\begin{align*}
(C_1 \parallel C_2) \parallel (C_1 \parallel C_3) &\refines C_1 \parallel (C_2 \parallel C_3) ,\\
C_1 \vchoice{p} (C_2 \vchoice{q} C_3) &\equiv (C_1 \vchoice{p} C_2) \vchoice{q} (C_1 \vchoice{p} C_3),\\
C_1 \hchoice{p} (C_2 \hchoice{q} C_3) &= (C_1 \hchoice{p} C_2) \hchoice{q} (C_1 \hchoice{p} C_3).
\end{align*}
\end{restatable}

\begin{restatable}[Distribution over different operators]{proposition}{resdistdiff}
\begin{align*}
C_1 \parallel (C_2 \vchoice{p} C_3)&\review{{\eqbij}} (C_1 \parallel C_2) \vchoice{p} (C_1 \parallel C_3), \\
C_1 \parallel (C_2 \hchoice{p} C_3) &= (C_1 \parallel C_2) \hchoice{p} (C_1 \parallel C_3),\\
C_1 \vchoice{p} (C_2 \hchoice{q} C_3) &= (C_1 \vchoice{p} C_2) \hchoice{q} (C_1 \vchoice{p} C_3).
\end{align*}
\end{restatable}

Unfortunately, the distribution of $\vchoiceop{p}$ over $\parallel$, $\hchoiceop{p}$ over $\parallel$, or 
$\hchoiceop{p}$ over $\vchoiceop{p}$ is not as well behaved.
\review{A complete discussion is avaiable in \version{the technical report\cite{arxiv}}{Appendix \ref{sec:further}}}

\subsection{Properties regarding cascading}

We conclude this section by exploring how our operators behave 
w.r.t. cascading (defined in Section~\ref{sec:preliminaries}).
Cascading of channels is fundamental in QIF, as it captures the concept of 
\review{a} system's \emph{post-processing} of another system's outputs, and it is also the key to the partial order on channels discussed in Section \ref{sec:preliminaries}.

\review{The next propositions explore whether it is possible to express a composition of two post-processed channels by a post-processing of their composition.}

\begin{restatable}{proposition}{rescascparallel}
Let $D_1 : \caly_1{\times}\calz_1 \rightarrow \reals$, $D_2 : \caly_2{\times}\calz_2 \rightarrow \reals$ be channels. Then,
\begin{equation*}
(C_1 D_1) \parallel (C_2 D_2)= (C_1 \parallel C_2) D^{\parallel},
\end{equation*}
where $D^{\parallel}: (\caly_1{\times}\caly_2){\times}(\calz_1{\times}\calz_2) \rightarrow \reals$ is defined, for all $y_1{\in}\caly_1$, $y_2{\in}\caly_2$, $z_1{\in}\calz_1$, and $z_2{\in}\calz_2$, as 
$D^{\parallel}((y_1,y_2),(z_1,z_2)) = D_1(y_1,z_1)D_2(y_2,z_2)$.
\end{restatable}

\begin{restatable}{proposition}{rescascvchoice}\label{lemma:cascvchoice}
Let $D_1 : \caly_1{\times}\calz_1 \rightarrow \reals$, $D_2 : \caly_2{\times}\calz_2 \rightarrow \reals$ be channels. Then,
\review{
$$
(C_1 D_1) \vchoice{p} (C_2 D_2)= (C_1 \vchoice{p} C_2) D^{\vchoiceop{}},
$$}
where $D^{\vchoiceop{}}{:}(\caly_1{\sqcup}\caly_2){\times}(\calz_1{\sqcup}\calz_2){\rightarrow}\reals$ is defined as
\review{$D^{\vchoiceop{}}((y,i), (z,j))=D_1(y,z)$ if $i{=}j{=}1$}, 
or $D_2(y,z)$ if $i=j=2$, or
$0$ otherwise,
for all $y_1{\in}\caly_1$, $y_2{\in}\caly_2$, $z_1{\in}\calz_1$, $z_2{\in}\calz_2$. 
\end{restatable}

A similar rule, however, does not hold for hidden choice. 
For example, let $C_1$ and $C_2$ be channels with input and output sets $\{1,2\}$, such that
$C_1 (x,x'){=}1$ if $x{=}x'$, or $0$ otherwise, and
$C_2 (x,x'){=}0$ if $x{=}x'$, or $1$ otherwise.
Let $D_1$ and $D_2$ be transparent channels whose output 
sets are disjoint. 
Then, $(C_1D_1){\hchoice{\nicefrac{1}{2}}}(C_2D_2)$ is a 
transparent channel, but $C_1{\hchoice{\nicefrac{1}{2}}}C_2$ 
is a null channel. 
Thus, it is impossible to describe
$(C_1D_1){\hchoice{\nicefrac{1}{2}}}(C_2D_2)$ as 
$C_1{\hchoice{\nicefrac{1}{2}}}C_2$ post-processed by some channel. 
However, we can establish a less general, yet relevant, equivalence.

\begin{restatable}{proposition}{rescaschchoice}
Let $C_1:\calx{\times}\caly{\rightarrow}\reals$ and $C_2: \calx{\times}\caly{\rightarrow}\reals$ be channels of the same type. Let $D: \caly{\times}\calz{\rightarrow}\reals$ be a channel. Then,
$
(C_1D) \hchoice{p} (C_2D)=(C_1 \hchoice{p}C_2)D.
$

\end{restatable}

\section{Information leakage of channel operators}
\label{sec:leakage}
This section presents the main contribution of our paper:
a series of results showing how, using the proposed
algebra, we can facilitate the security analysis 
of compound systems.
Our results are given in terms \review{of} the 
$g$-leakage framework introduced in Section~\ref{sec:preliminaries}, and we focus on two central
problems.
\review{ For the remaining of the section, let $C_1: \calx{\times}\caly_1 \rightarrow \reals$ and $C_2: \calx{\times}\caly_2 \rightarrow \reals$ be compatible channels. }

\subsection{The problem of compositional vulnerability}
\review{ 
The first problem consists in estimating 
the information leakage of a compound system
in terms of the leakage of its components.
This is formalized as follows.}

\review{\textbf{The problem of compositional vulnerability:}
\emph{Given a composition operator $\ast$ on channels,
a prior $\pi{\in}\dist\calx$, 
and a gain function $g$, how can we estimate $V_g \hyperc{\pi}{C_1 \ast C_2}$ in terms of $V_g \hyperc{\pi}{C_1}$ and $V_g \hyperc{\pi}{C_2}$?}
}
\begin{restatable}[Upper and lower bounds for $V_g$ w.r.t. $\parallel$ ]{theorem}{resuplowparallel}  \label{theorem:parineq} \review{For all gain functions $g$ and $\pi \in \dist\calx$}, let $\calx'= \{x{\in}\calx \mid \exists w{\in}\calw$ s.t. $\pi(x)g(w,x)>0\}$. Then
 \begin{align*}
  V_g \hyperc{\pi}{C_1{\parallel}C_2} &{\geq} \max (V_g \hyperc{\pi}{C_1} ,V_g \hyperc{\pi}{C_2}), \qquad \text{and}\\
  V_g \hyperc{\pi}{C_1{\parallel}C_2} &{\leq} {\min}\left( V_g \hyperc{\pi}{C_1} \sum_{y_2} \max\limits_{x \in \calx'} C_2(x,y_2) ,  V_g \hyperc{\pi}{C_2} \sum_{y_1} \max\limits_{x \in \calx'} C_1(x,y_1) \right).
 \end{align*}
\end{restatable}
\begin{restatable}[Linearity of $V_g$ w.r.t. $\vchoiceop{p}$]{theorem}{reseqvchoice}
\review{For all gain functions $g$, $\pi \in \dist\calx$ and $p \in [0,1]$,}
\label{theorem:reseqvchoice}
\begin{equation*}
V_g\hyperc{\pi}{C_1 \vchoice{p} C_2}=p  V_g\hyperc{\pi}{C_1}+(1-p)  V_g\hyperc{\pi}{C_2} .
\end{equation*}
\end{restatable}
\begin{restatable}[Upper and lower bounds for $V_g$ w.r.t. $\hchoiceop{p}$]{theorem}{resuplowhchoice} \label{theorem:hineq}
 \review{For all gain functions $g$, $\pi \in \dist\calx$ and $p \in [0,1]$,}
 \begin{align*}
 V_g\hyperc{\pi}{C_1 \hchoice{p} C_2} &\geq \max(p  V_g\hyperc{\pi}{C_1}, (1-p)  V_g\hyperc{\pi}{C_2}), \qquad \text{and}\\
 V _g\hyperc{\pi}{C_1 \hchoice{p} C_2} &\leq p 
  V_g\hyperc{\pi}{C_1}+(1-p)  V_g\hyperc{\pi}{C_2}.
 \end{align*}
\end{restatable}

The three theorems  above yield an interesting order between the operators

\begin{restatable}[Ordering between operators]{corollary}{resoperorder}
\label{cor:ordering}
Let \review{$\pi \in \dist \calx$}, $g$ be a gain function and $p \in [0,1]$. Then
$V_g\hyperc{\pi}{C_1 \parallel C_2} \geq V_g\hyperc{\pi}{C_1 \vchoice{p} C_2} \geq V_g\hyperc{\pi}{C_1 \hchoice{p} C_2}$.
\end{restatable}

\subsection{The problem of relative monotonicity}
\review{The second problem concerns establishing whether
a component channel of a larger system can be safely
substituted with another component, i.e.,
whether substituting a component with another can cause an increase in the information leakage of the system as a whole.
This is formalized as follows.}

\review{\textbf{The problem of relative monotonicity:}
\emph{Given a composition operator $\ast$ on channels,
a prior $\pi \in \dist \calx$, 
and a gain function $g$, is it the case that 
$V_g\hyperc{\pi}{C_1} \leq V_g\hyperc{\pi}{C_2} \Leftrightarrow \forall C \in \mathcal{C}_\calx. \ V_g\hyperc{\pi}{C_1 \ast C} \leq V_g\hyperc{\pi}{C_2 \ast C}\ ?$}
}

We start by showing that relative monotonicity holds 
for visible choice.
Note, however, that because
$V_g\hyperc{\pi}{C_1 \vchoice{p} C} \leq V_g\hyperc{\pi}{C_2 \vchoice{p} C}$ is vacuously true
if $p = 0$, we consider only $p \in (0,1]$.

\begin{restatable}[Relative monotonicity for $\vchoiceop{p}$]{theorem}{ressecvchoice}\review{For all gain functions $g$, $\pi \in \dist\calx$ and $p \in (0,1]$,}
\begin{equation*}
V_g\hyperc{\pi}{C_1}\leq V_g\hyperc{\pi}{C_2} \Leftrightarrow \forall C.  \: V_g\hyperc{\pi}{C_1 \vchoice{p} C} \leq V_g\hyperc{\pi}{C_2 \vchoice{p} C}.
\end{equation*}
\end{restatable}

Interestingly, relative monotonicity does not hold for 
the parallel operator.
This means that the fact that a channel $C_{1}$ is always
more secure than a channel $C_{2}$ does not guarantee that
if we replace $C_{1}$ for $C_{2}$ in a parallel context 
we necessarily obtain a more secure system.\footnote{As a counter-example, consider channels
$C_{1}=
\left(\begin{smallmatrix}
1 & 0 \\
1 & 0 \\
0 & 1 
\end{smallmatrix}\right)
$
and 
$C_{2}=
\left( \begin{smallmatrix}
1 & 0 \\
0 & 1 \\
0 & 1 
\end{smallmatrix} \right)$. Let $\pi_u = \{\nicefrac{1}{3}, \nicefrac{1}{3}, \nicefrac{1}{3}\}$ and $g_{id}: \calx \times \calx \rightarrow [0,1]$ s.t. $g_{id}(x_1,x_2)=1$ if $x_1=x_2$ and $0$ otherwise. Then, $V_{g_{id}}\hyperc{\pi_u}{C_2} \leq V_{g_{id}}\hyperc{\pi_u}{C_1}$, but $V_{g_{id}}\hyperc{\pi_u}{C_2 \parallel C_1} > V_{g_{id}}\hyperc{\pi_u}{C_1 \parallel C_1}$.
}
However, when the adversary's knowledge (represented by the prior $\pi$)
or preferences (represented by the gain-function $g$) are known, we can obtain
a constrained result on leakage by fixing only 
$\pi$ or $g$.

\begin{restatable}[Relative monotonicity for $\parallel$]{theorem}{ressecparallel} \label{theorem:parord}
\review{For all gain functions $g$ and $\pi \in \dist\calx$ }
\begin{align*}
\forall \pi' . \; V_g\hyperc{\pi'}{C_1} \leq V_g\hyperc{\pi'}{C_2}  &\Leftrightarrow  \forall \pi' ,C. \; V_g\hyperc{\pi'}{C_1 \parallel C} \leq V_g\hyperc{\pi'}{C_2 \parallel C}, \quad \text{and} \\
\forall g'. \; V_{g'}\hyperc{\pi}{C_1} \leq V_{g'}\hyperc{\pi}{C_2}  &\Leftrightarrow  \forall g',C. \; V_{g'}\hyperc{\pi}{C_1 \parallel C} \leq V_{g'}\hyperc{\pi}{C_2 \parallel C}.
\end{align*}
\end{restatable}

Perhaps surprisingly, hidden choice does not respect
relative monotonicity, even 
when we only consider channels that respect the refinement relation introduced in section \ref{sec:preliminaries}.

\begin{restatable}[Relative monotonicity for $\hchoiceop{p}$]{theorem}{ressechchoicei}\label{theorem:ressechchoicei}
For all $p{\in}(0,1)$, there are $C_1{:} \calx{\times}\caly{\rightarrow}\reals$ and $C_2{:}\calx{\times}\caly{\rightarrow} \reals$ such that
 \begin{align*}
\forall \pi, g. \: V_{g}\hyperc{\pi}{C_1}{\leq} V_{g}\hyperc{\pi}{C_2} &\mbox{ and } \exists \pi', g', C. \:V_{g'}\hyperc{\pi'}{C_1 \hchoice{p} C} {>} V_{g'}\hyperc{\pi'}{C_2 \hchoice{p} C} ,
\end{align*}
\end{restatable}

The converse, however, is true.

\begin{restatable}[Relative monotonicity for $\hchoiceop{p}$, cont.]{theorem}{ressechchoiceii}
\review{For all gain functions $g$, $\pi \in \dist\calx$ and $p \in (0,1]$,}
\begin{equation*}
     \forall C. \:V_g\hyperc{\pi}{C_1 \hchoice{p} C} \leq V_g\hyperc{\pi}{C_2 \hchoice{p} C} \Rightarrow  V_g\hyperc{\pi}{C_1}\leq V_g\hyperc{\pi}{C_2} .
\end{equation*}
\end{restatable}

\section{Case study: the Crowds protocol}
\label{sec:casestudy}
In this section we apply the theoretical techniques
developed in this paper to the well-known \emph{Crowds} anonymity protocol~\cite{crowds}.
Crowds was designed to protect the identity
of a group of \emph{users} who wish to anonymously 
send requests to a \emph{server}, and it is the basis of the
widely used protocols \emph{Onion Routing}~\cite{Goldschlag:96:IH} and \emph{Tor}~\cite{Dingledine:04:USENIX}.

The protocol works as follows.
When a user wants to send a request to the server, he first randomly
picks another user in the group and forwards the request to that user.
From that point on, each user, upon receiving a request from another user, 
sends it to the server with probability $p\review{\in (0,1]}$, or forwards it to another user 
with probability $1{-}p$.
This second phase repeats until the message reaches the server.

It is assumed that the adversary controls the server and some
\emph{corrupt} users among the regular, \emph{honest}, ones.
When a corrupt user receives a forwarded request, he shares the forwarder's
identity with the server, and we say that the forwarder was \emph{detected}. 
As no information can be gained after a corrupt user intercepts a 
request, we need only consider the protocol's execution until a detection
occurs, or the message reaches the server.

In Crowds' original description, all users have equal probability
of being forwarded a message, regardless of the forwarder.
The channel modeling such a case is easily computed, and well-known
in the literature.
Here we consider the more general case in which each user may employ a different probability 
distribution when choosing which user to forward a request to.
Thus, we can capture scenarios in which not all users can easily 
reach each other (a common problem in, for instance, \emph{ad-hoc}
networks).
We make the simplifying assumption that corrupt users are evenly distributed, i.e., that all honest users have the same probability $q{\in}\review{(0,1]}$ of 
choosing a corrupt user to forward a request to.


We model Crowds as a channel $Crowds{:}\calx{\times}\caly{\rightarrow} \reals$.
The channel's input, taken from set $\mathcal{\calx}{=}\{u_1, u_2,\ldots, u_{n_c}\}$,
represents the identity $u_{i}$ of the honest user (among a total of $n_c$ honest users) who initiated the request. 
The channel's output is either the identity of a detect user---i.e., a value from $\mathcal{D}{=} \{d_1,d_2,\ldots,d_{n_c}\}$,
where where $d_i$ indicates user $u_i$ was detected---or the identity of a user who forwarded the message to the server---i.e., a value from $\mathcal{S}{=} \{s_1,s_2,\ldots,s_{n_c}\}$,
where $s_i$ indicates user $u_i$ forwarded a message to the server.
Note that $\mathcal{D}$ and $\mathcal{S}$ are disjoint, and the channel's
output set is $\caly = \mathcal{D} \cup \mathcal{S}$.

To compute the channel's entries, we model the protocol as a time-stationary Markov chain $M=(\mathcal{U}, \boldsymbol{P})$, where the set of states is
the set of honest users $\mathcal{U}$, and its transition function 
is such that $\boldsymbol{P}(u_i,u_j)$ is the probability of $u_j$
being the recipient of a request forwarded by $u_i$,
given that $u_i$ will not be detected.

We then define four auxiliary channels.
Transparent channels
 $I_d{:}\mathcal{U}{\times}\mathcal{D} {\rightarrow}\reals$ and $I_s{:}\mathcal{U}{\times}\mathcal{S}{\rightarrow}\reals$
are defined as 
$I_d(u_i,d_j){=}1$ if $i{=}j$, or $0$ otherwise, and
$I_s(u_i,s_j)=1$ if $i{=}j$, or $0$ otherwise; and two other channels 
$P_d{:}\mathcal{D}{\times}\mathcal{D}{\rightarrow}\reals$ and 
$P_s{:}\mathcal{S}{\times}\mathcal{S}{\rightarrow}\reals$, 
based on our Markov chain $M$,
are defined as 
$P_d (d_i,d_j){=}P_s (s_i,s_j){=}\boldsymbol{P}(u_i,u_j)$.

We begin by reasoning about what happens if each request can be 
forwarded only once.
There are two possible situations: either the initiator
is detected, or he forwards the request to an honest user,
who will in turn send it to the server.
The channel corresponding to the initiator being detected is 
$I_d$, since in this case the output has to be 
$d_i$ whenever $u_i$ is the initiator. 
The channel corresponding to the latter situation is $I_sP_s$---i.e., 
the channel $I_s$ postproccessed by $P_s$. 
This is because, being $P_s$ based on the transition function of $M$,
the entry $(I_sP_s) (u_i,s_j)$ gives us exactly the probability that user
$u_j$ received the request originated by user $u_i$ after it being
forwarded once.
Therefore, when Crowds is limited to one forwarding, it
can be modeled by the channel $I_d \hchoice{q}I_sP_s$~\footnote{To simplify notation, we assume cascading has precedence over hidden choice, i.e., $AB{\hchoice{p}}CD = (AB) {\hchoice{p}}(CD)$.},
representing the fact that: 
(1) with probability $q$ the initiator is detected, 
and the output is generated by $I_d$; and
(2) with probability $1-q$ the output is generated by $I_sP_s$.

Let us now cap our protocol to at most two forwards.
If the initiator is not immediately detected, the 
first recipient will have a probability $p$ of sending the message
to the server.
If the recipient forwards the message instead, he may be detected. 
Because the request was already forwarded once, the channel
that will produce the output in this case 
is $I_dP_d$ (notice that, despite this channel being equivalent to 
$I_sP_s$, it is of a different type).
On the other hand, if the first recipient forwards the message to 
an honest user, this second recipient will now send the message to 
the server, making the protocol produce an output according to 
$I_sP_sP_s$ (or simply $I_sP_s^2$), since $(I_sP_s^2) (u_i,s_j)$ 
is the probability that user $u_j$ received the request originated 
by user $u_i$ after it being forwarded twice.
Therefore, when Crowds is limited to two forwardings, it can be modeled
by the channel $I_d{\hchoice{q}}(I_sP_s{\hchoice{p}}(I_dP_d {\hchoice{q}}I_sP_s^2))$.
Note the disposition of the parenthesis reflects the order in which the 
events occur. 
First, there is a probability $q$ of the initiator being detected,
and $1-q$ of the protocol continuing. 
Then, there is a probability $p$ of the first recipient sending it 
to the server, and so on.

Proceeding this way, we can inductively construct a sequence 
 $\{C_i\}_{i \in \mathbb{N^*}}$, 
 \commentA{Is $\mathbb{N}^*$ standard notation for naturals without $0$?}
\begin{equation*}
C_{i}=I_d \hchoice{q}(I_sP_s \hchoice{p_{}} (I_dP_d \hchoice{q}(\ldots\hchoice{p_{}}(I_dP_d^{i-1} \hchoice{q}I_sP_s^{i} )\ldots))),
\end{equation*}
in which each $C_i$ represents our protocol capped at
$i$ forwards per request.
We can then obtain $Crowds$ by taking $\lim_{i \rightarrow \infty} C_i$.
From that, Theorem~\ref{theorem:hineq} and Proposition~\ref{prop:associative}, we can derive the following bounds
on the information leakage of Crowds.

\begin{restatable}{theorem}{rescrowdsbounds}\label{theorem:bound}
Let $\{t_i\}_{i \in \mathbb{N}}$ be the sequence in which
$t_{2i}{=}1{-}(1{-}q)^{i{+}1}(1{-}p)^i$ and  
$t_{(2i{+}1)}=1{-}(1-q)^{i{+}1}(1{-}p)^{i{+}1}$ for all $i{\in}\mathbb{N}$.

Let $K_m{=}((\ldots(I_d \hchoice{\nicefrac{t_0}{t_1}} I_sP_s)\hchoice{\nicefrac{t_1}{t_2}}\ldots)\hchoice{\nicefrac{t_{2m{-}1}}{t_{2m}}}(I_dP_d^{m})$. Then, $\forall m{\in}\mathbb{N}^*$,
\begin{align}
&V_g \hyperc{\pi}{\lim\limits_{i \rightarrow \infty} C_i} ~\geq~ t_{2m}V_g\hyperc{\pi}{K_m}, \label{eq:crowdslower}\\
&V_g \hyperc{\pi}{\lim\limits_{i \rightarrow \infty} C_i} ~\leq~ \review{
t_{2m}V_g\hyperc{\pi}{K_m} + (1-t_{2m})V_g\hyperc{\pi}{I_sP_s^{m+1}},} \quad \text{and} \label{eq:crowdsupper}\\
&(1-t_{2m})V_g\hyperc{\pi}{I_sP_s^{m+1}} ~\leq~ (1-q)^{m+1}(1-p_{})^{m}. \label{eq:approx}
\end{align}
\end{restatable}

Equations~\eqref{eq:crowdslower} and \eqref{eq:crowdsupper} 
provide an effective way to approximate the $g$-leakage of
information of the channel $Crowds$ with arbitrary precision, 
whereas Equation~\eqref{eq:approx} lets us estimate how many 
interactions are needed for that.

\review{To obtain $K_m$, we need to calculate $m$ matrix multiplications, which surpass the cost of
computing the $m$ hidden choices
(which are only matrix additions).
Thus, Theorem~\ref{theorem:bound}
implies we can obtain a channel whose posterior vulnerability differs from that of $Crowds$ by at most $(1{-}q)^{m{+}1}(1{-}p^{m})$ in ${\equiv}O(mn_c^{2.807})$ time} (using the Strassen algorithm for matrix multiplication~\cite{Strassen:69:NM}).
Since $p$ is typically high,  
$(1{-}q)^{m{+}1}(1{-}p_{})^{m}$ decreases very fast.
For instance, for a precision of $0.001$ on the leakage bound, we need $m{=}10$ when  $(1{-}q)(1{-}p_{})$ is $0.5$, $m{=}20$ when it is $0.7$, and $m{=}66$ when it is $0.9$, 
regardless of the number $n_c$ of honest users.


Therefore, our method has time complexity $O(n_c^{2.807})$ 
when the number of users is large (which is the usual case for Crowds),
and reasonable values of forward probability $p_{}$, and precision.
To the best of our knowledge this method is the fastest in
the literature, beating the previous $O(n_c^{3.807})$
that can be achieved by modifying the method presented in \cite{Andres:10:TACAS}---although their method does not require 
our assumption of corrupt users being evenly distributed.

\section{Related work}
\label{sec:related}
Compositionality is a \review{fundamental} notion in computer science,
and it has been subject of growing interest in the QIF community.

Espinoza and Smith \cite{Espinoza:13:IC} derived a number of
\emph{min-capacity} bounds for different channel 
compositions, including cascading and parallel composition.

However, it was not until recently that compositionality results
regarding the more general metrics of $g$-leakage started to be explored.
Kawamoto et al. \cite{Kawamoto:17:LMCS} defined a generalization of the 
parallel operator for channels with different input sets, and 
gave upper bounds for the corresponding information leakage.
Our bounds for compatible channels (Theorem~\ref{theorem:parineq}) are tighter than theirs.

Recently, Engelhardt \cite{Engelhardt:17:ESORICS} 
defined the \emph{mix operator}, another generalization of parallel composition,
and derived results similar to ours regarding the parallel operator.
Specifically, he provided commutative and associative properties (Propositions \ref{prop:commutative} and \ref{prop:associative}), and from his results the 
lower bound of Theorem~\ref{theorem:parineq} can be inferred. 
He also proved properties similar to the ones in Proposition \ref{prop:null}, albeit using more restrictive definitions of null
and transparent channels. 

Both Kawamoto et al. and Engelhardt provided results similar to Theorem~\ref{theorem:parord},
but ours is not restricted to when one channel is refined by the other.

\review{Just recently, Alvim et. al investigated algebraic properties of hidden and visible choice operators in the context of game-theoretic aspects of QIF \cite{Alvim:18:POST}, and derived the upper bounds of Theorems \ref{theorem:reseqvchoice} and \ref{theorem:hineq}.
Here we expanded the algebra to the interaction among operators, including parallel composition, derived more comprehensive bounds on their leakage, and applied our results to the Crowds protocol}.

\section{Conclusions and future work}
\label{sec:conc}
\review{In this paper we proposed an algebra to express numerous
component compositions in systems that arise from typical ways in components interact in practical
scenarios.
We provided fundamental algebraic properties of these operators,
and studied several of their leakage properties.
In particular, we obtained new results regarding their motonicity properties and stricter bounds for the parallel and hidden choice operators.}
These results are of practical interest for the QIF community,
as they provide helpful tools for modeling large systems and
analyzing their security properties.

The list of operators we explored in this paper, however, does not seem to 
capture every possible interaction of components of real systems.
As future work we wish to find other operators and increase the 
expressiveness of our approach.

\paragraph*{Acknowledgments}
\begin{small}
Arthur Am\'{e}rico and M\'{a}rio S. Alvim are supported by CNPq, CAPES, and FAPEMIG.
Annabelle McIver is supported by ARC grant DP140101119.
\end{small}

\bibliographystyle{splncs}
\bibliography{references,short}
\version{}{
\newpage
\appendix
\section{Proofs of technical results}
\label{sec:proofs}
In this section we provide proofs for our technical results.

\subsection{Proofs of Section~\ref{sec:algebraic}}
In this section, we will consider $C_1: \calx \times \caly_1 \rightarrow \reals$, $C_2: \calx \times \caly_2 \rightarrow \reals$, 
and $C_3: \calx \times \caly_3 \rightarrow \reals$ to be compatible channels.
\rescommutative*

\begin{proof}
$\bullet$ ($\parallel$) The bijection is given by $\review{\psi} ((y_1,y_2))=(y_2,y_1)$. For all $x \in \calx$, $y_1 \in \caly_1$ and $y_2 \in \caly_2$,
\begin{align*}
&(C_1 \parallel C_2 )(x,(y_1, y_2))\\
=&C_1(x, y_1)C_2(x,y_2) &\text{(by def. of $\parallel$)}\\
=&(C_2 \parallel C_1 )(x,(y_2, y_1)) &\text{(by def. of $\parallel$)}
\end{align*}

$\bullet$ ($\vchoiceop{p}$) The bijection is given by $\review{\psi} ((y,1))=(y,2)$ and $\review{\psi}((y,2))=(y,1)$.
For all $x \in \calx$ and $y \in \caly_1$,
\begin{align*}
&(C_1 \vchoice{p} C_2 )(x,(y, 1))\\
=&pC_1(x, y) &\text{(by def. of $\vchoice{p}$)}\\
=&(C_2 \vchoice{(1-p)} C_1 )(x,(y, 2)) &\text{(by def. of $\vchoice{(1-p)}$)}
\end{align*}

Similarly, for all $x \in \calx$ and $y \in \caly_2$
\begin{align*}
&(C_1 \vchoice{p} C_2 )(x,(y, 2))\\
=&(1-p)C_2(x, y) &\text{(by def. of $\vchoice{p}$)}\\
=&(C_2 \vchoice{(1-p)} C_1 )(x,(y, 1))&\text{(by def. of $\vchoice{(1-p)}$)}
\end{align*}

$\bullet$ ($\hchoiceop{p}$) For all $x \in \calx$ and $y \in \caly_1 \cap \caly_2$ ,
\begin{align*}
&(C_1 \hchoice{p} C_2 )(x,y)\\
=&pC_1(x, y) + (1-p)C_2(x, y) &\text{(by def. of $\hchoice{p}$)}\\
=&(C_2 \hchoice{(1-p)} C_1 )(x,y) &\text{(by def. of $\hchoice{(1-p)}$)}
\end{align*}

For all $x \in \calx$ and $y \in \caly_1 \setminus \caly_2$
\begin{align*}
&(C_1 \hchoice{p} C_2 )(x,y)\\
=&pC_1(x, y)&\text{(by def. of $\hchoice{p}$)}\\
=&(C_2 \hchoice{(1-p)} C_1 )(x,y) &\text{(by def. of $\hchoice{(1-p)}$)}
\end{align*}

Finally, for all $x \in \calx$ and $y \in \caly_2 \setminus \caly_1$
\begin{align*}
&(C_1 \hchoice{p} C_2 )(x,y)\\
=&(1-p)C_2(x, y) &\text{(by def. of $\hchoice{p}$)}\\
=&(C_2 \hchoice{(1-p)} C_1 )(x,y) &\text{(by def. of $\hchoice{(1-p)}$)}
\end{align*}
\qed
\end{proof}

\resassociative*

\begin{proof}
$\bullet$ ($\parallel$) The bijection is given by $\review{\psi}(((y_1,y_2), y_3))=(y_1,(y_2, y_3))$.
For all $x \in \calx$, $y_1 \in \caly_1$, $y_2 \in \caly_2$ 
and $y_3 \in \caly_3$,
\begin{align*}
&((C_1 \parallel C_2) \parallel C_3 )(x,((y_1, y_2), y_3))\\
=&(C_1 \parallel C_2)(x, (y_1, y_2))C_3(x,y_3) &\text{(by def. of $\parallel$)}\\
=&C_1(x,y_1)C_2(x, y_2)C_3(x,y_3) &\text{(by def. of $\parallel$)}\\
=&C_1(x,y_1)(C_2 \parallel C_3 )(x,(y_2, y_3)) &\text{(by def. of $\parallel$)}\\
=&(C_1 \parallel (C_2 \parallel C_3) )(x,(y_1, (y_2, y_3))) &\text{(by def. of $\parallel$)}
\end{align*} 

$\bullet$ ($\vchoiceop{p}$) $\review{\psi}$ is defined as
\begin{align*}
\review{\psi}(((y_1, 1), 1))&=(y_1,1)\\
\review{\psi}(((y_2, 2), 1))&=((y_2, 1),2)\\
\review{\psi}((y_3, 2))&=((y_3, 2), 1)
\end{align*}
for all $y_1 \in \caly_1$, $y_2 \in \caly_2$ and $y_3 \in \caly_3$. For all $x \in \calx$ and $y \in \caly_1$ ,
\begin{align*}
&((C_1 \vchoice{p} C_2) \vchoice{q} C_3 )(x,((y, 1),1))\\
=&pqC_1(x, y) &\text{(by def. of $\vchoice{p}$, $\vchoice{q}$)}\\
=&p'C_1(x, y) &\text{(by def. of $p'$)}\\
=&(C_1 \vchoice{p'} (C_2 \vchoice{q'} C_3) )(x,(y,1)) &\text{(by def. of $\vchoice{p'}$, $\vchoice{q'}$)}
\end{align*}

For all $x \in \calx$ and $y \in \caly_2$ ,
\begin{align*}
&((C_1 \vchoice{p} C_2) \vchoice{q} C_3 )(x,((y, 2),1))\\
=&(1-p)qC_2(x, y) &\text{(by def. of $\vchoice{p}$, $\vchoice{q}$)}\\
=&(1-p')q'C_2(x, y) &\text{(by def. of $p'$, $q'$)}\\
=&(C_1 \vchoice{p'} (C_2 \vchoice{q'} C_3) )(x,((y,1), 2)) &\text{(by def. of $\vchoice{p'}$, $\vchoice{q'}$)}
\end{align*}

For all $x \in \calx$ and $y \in \caly_3$ ,
\begin{align*}
&((C_1 \vchoice{p} C_2) \vchoice{q} C_3 )(x,(y,2))\\
=&(1-q)C_3(x, y) &\text{(by def. of $\vchoice{p}$, $\vchoice{q}$)}\\
=&(1-p')(1-q')C_3(x, y) &\text{(by def. of $p'$, $q'$)}\\
=&(C_1 \vchoice{p'} (C_2 \vchoice{q'} C_3) )(x,((y,2), 2)) &\text{(by def. of $\vchoice{p'}$, $\vchoice{q'}$)}
\end{align*}

$\bullet$ ($\hchoiceop{p}$) For all $x \in \calx$ and $y \in \caly_1 \cap \caly_2 \cap \caly_3$ ,
\begin{align*}
&((C_1 \hchoice{p} C_2) \hchoice{q} C_3 )(x,y)\\
=&pqC_1(x, y) +(1{-}p)qC_2(x,y){+}(1{-}q)C_3(x, y)&\text{(by def. of $\hchoice{p}$, $\hchoice{q}$)}\\
=&p'C_1(x, y) {+}(1{-}p')q'C_2(x,y){+}(1{-}p')(1{-}q')C_3(x, y) &\text{(by def. of $p'$, $q'$)}\\
=&(C_1 \hchoice{p'} (C_2 \hchoice{q'} C_3) )(x, y) &\text{(by def. of $\hchoice{p'}$, $\hchoice{q'}$)}
\end{align*}

For all $x \in \calx$ and $y \in (\caly_1 \cap \caly_2) \setminus \caly_3$ ,
\begin{align*}
&((C_1 \hchoice{p} C_2) \hchoice{q} C_3 )(x,y)\\
=&pqC_1(x, y) +(1-p)qC_2(x,y) &\text{(by def. of $\hchoice{p}$, $\hchoice{q}$)}\\
=&p'C_1(x, y) +(1-p')q'C_2(x,y)&\text{(by def. of $p'$, $q'$)}\\
=&(C_1 \hchoice{p'} (C_2 \hchoice{q'} C_3) )(x, y) &\text{(by def. of $\hchoice{p'}$, $\hchoice{q'}$)}
\end{align*}

For all $x \in \calx$ and $y \in (\caly_1 \cap \caly_3 )\setminus \caly_2$ ,
\begin{align*}
&((C_1 \hchoice{p} C_2) \hchoice{q} C_3 )(x,y)\\
=&pqC_1(x, y) +(1-q)C_3(x, y)&\text{(by def. of $\hchoice{p}$, $\hchoice{q}$)}\\
=&p'C_1(x, y) +(1-p')(1-q')C_3(x, y) &\text{(by def. of $p'$, $q'$)}\\
=&(C_1 \hchoice{p'} (C_2 \hchoice{q'} C_3) )(x, y) &\text{(by def. of $\hchoice{p'}$, $\hchoice{q'}$)}
\end{align*}

For all $x \in \calx$ and $(y \in \caly_2 \cap \caly_3 )\setminus \caly_1$ ,
\begin{align*}
&((C_1 \hchoice{p} C_2) \hchoice{q} C_3 )(x,y)\\
=&(1-p)qC_2(x,y)+(1-q)C_3(x, y)&\text{(by def. of $\hchoice{p}$, $\hchoice{q}$)}\\
=&(1-p')q'C_2(x,y)+(1-p')(1-q')C_3(x, y) &\text{(by def. of $p'$, $q'$)}\\
=&(C_1 \hchoice{p'} (C_2 \hchoice{q'} C_3) )(x, y) &\text{(by def. of $\hchoice{p'}$, $\hchoice{q'}$)}
\end{align*}

For all $x \in \calx$ and $y \in \caly_1 \setminus( \caly_2 \cup \caly_3)$ ,
\begin{align*}
&((C_1 \hchoice{p} C_2) \hchoice{q} C_3 )(x,y)\\
=&pqC_1(x, y)&\text{(by def. of $\hchoice{p}$, $\hchoice{q}$)}\\
=&p'C_1(x, y) &\text{(by def. of $p'$)}\\
=&(C_1 \hchoice{p'} (C_2 \hchoice{q'} C_3) )(x, y) &\text{(by def. of $\hchoice{p'}$, $\hchoice{q'}$)}
\end{align*}

For all $x \in \calx$ and $y \in \caly_2 \setminus ( \caly_1 \cup \caly_3)$ ,
\begin{align*}
&((C_1 \hchoice{p} C_2) \hchoice{q} C_3 )(x,y)\\
=&(1-p)qC_2(x,y) &\text{(by def. of $\hchoice{p}$, $\hchoice{q}$)}\\
=&(1-p')q'C_2(x,y) &\text{(by def. of $p'$, $q'$)}\\
=&(C_1 \hchoice{p'} (C_2 \hchoice{q'} C_3) )(x, y) &\text{(by def. of $\hchoice{p'}$, $\hchoice{q'}$)}
\end{align*}

For all $x \in \calx$ and $y \in \caly_3 \setminus (\caly_1 \cup \caly_2)$ ,
\begin{align*}
&((C_1 \hchoice{p} C_2) \hchoice{q} C_3 )(x,y)\\
=&(1-q)C_3(x, y)&\text{(by def. of $\hchoice{p}$, $\hchoice{q}$)}\\
=&(1-p')(1-q')C_3(x, y) &\text{(by def. of $p'$, $q'$)}\\
=&(C_1 \hchoice{p'} (C_2 \hchoice{q'} C_3) )(x, y) &\text{(by def. of $\hchoice{p'}$, $\hchoice{q'}$)}
\end{align*}
\qed
\end{proof}

\resnull*
\begin{proof}
$\bullet$ (null channel, $\parallel$) Firstly, for any null channel $\nullchannel: \calx \times \calz \rightarrow \reals$, we notice that
$\nullchannel(x_1,z)=\nullchannel(x_2,z)$ for any $x_1,x_2 \in \calx$ and $z \in \calz$. 
Thus, given $z \in \calz$ we uniquely define $\nullchannel(z)=\nullchannel(x,z) \text{ for any } x \in \calx$.

We then have that, for any $\pi$ and $g$:
\begin{align*}
 &V_{g}\hyperc{\pi}{C_1 \parallel \nullchannel}\\
 =&\sum_{z \in \mathcal{Z}} \sum_{y \in \mathcal{Y}} \max\limits_{w \in \mathcal{W}} \sum_{x \in \calx} C_1(x,y) \cdot \nullchannel(x,z) \cdot g(w,x) \cdot \pi(x) &\text{(by def. of $\parallel$)}\\
 =&\sum_{z \in \mathcal{Z}} \sum_{y \in \mathcal{Y}} \max\limits_{w \in \mathcal{W}} \sum_{x \in \calx} C_1(x,y) \cdot \nullchannel(z) \cdot g(w,x) \cdot \pi(x) &\text{(by def. of $\nullchannel(z)$)}\\
 =&\sum_{z \in \mathcal{Z}} \nullchannel(z) \sum_{y \in \mathcal{Y}} \max\limits_{w \in \mathcal{W}} \sum_{x \in \calx} C_1(x,y) \cdot g(w,x) \cdot \pi(x) &\text{($\nullchannel(z)$ is constant given $z$)}\\
 =&\sum_{z \in \mathcal{Z}} \nullchannel(z) \cdot V_{g}\hyperc{\pi}{C_1} & \text{(by def. of vulnerability)}\\
 =&V_{g}\hyperc{\pi}{C_1} &\text{($\nullchannel$ is a channel)}\\ 
\end{align*}

$\bullet$ (null channel, $\vchoiceop{p}$)
We have, for all $\pi \in \dist \calx$, and all gain functions $g$,
\begin{align*}
&V_g \hyperc{\pi}{C_1 \vchoice{p}\nullchannel}\\
=&pV_g \hyperc{\pi}{C_1}+(1-p)V_g \hyperc{\pi}{\nullchannel} &\text{(from theorem \ref{theorem:reseqvchoice})}\\
\leq& pV_g \hyperc{\pi}{C_1}+(1-p)V_g \hyperc{\pi}{C_1} &\text{($V_g \hyperc{\pi}{\nullchannel} \leq V_g \hyperc{\pi}{C_1}$)}\\
=& V_g \hyperc{\pi}{C_1} 
\end{align*}

$\bullet$ (null channel, $\hchoiceop{p}$)
We have, for all $\pi \in \dist \calx$, and all gain functions $g$,
\begin{align*}
&V_g \hyperc{\pi}{C_1 \hchoice{p}\nullchannel}\\
\leq&pV_g \hyperc{\pi}{C_1}+(1-p)V_g \hyperc{\pi}{\nullchannel} &\text{(from theorem \ref{theorem:hineq})}\\
\leq& pV_g \hyperc{\pi}{C_1}+(1-p)V_g \hyperc{\pi}{C_1} &\text{($V_g \hyperc{\pi}{\nullchannel} \leq V_g \hyperc{\pi}{C_1}$)}\\
=& V_g \hyperc{\pi}{C_1} 
\end{align*}

$\bullet$ (transparant channel, $\parallel$) From theorem \ref{theorem:parineq}, we have, for all $\pi \in \dist \calx$
and gain functions $g$,
$V_g\hyperc{\pi}{C_1 \parallel \transparentchannel} \geq  V_g\hyperc{\pi}{\transparentchannel}$.
A transparent channel refines any other compatible channel. Thus, $C_1 \parallel \transparentchannel \equiv \transparentchannel$\\

$\bullet$ (transparent channel, $\vchoiceop{p}$) We have, for all $\pi \in \dist \calx$ and gain functions $g$,
\begin{align*}
&V_g \hyperc{\pi}{C_1 \vchoice{p}\transparentchannel}\\
=&pV_g \hyperc{\pi}{C_1}+(1-p)V_g \hyperc{\pi}{\transparentchannel} &\text{(from theorem \ref{theorem:reseqvchoice})}\\
\geq& pV_g \hyperc{\pi}{C_1}+(1-p)V_g \hyperc{\pi}{C_1} &\text{($V_g \hyperc{\pi}{\transparentchannel} \geq V_g \hyperc{\pi}{C_1}$)}\\
=& V_g \hyperc{\pi}{C_1} 
\end{align*}
\qed
\end{proof}
\residempotency*

\begin{proof}
$\bullet$ ($\parallel$) From theorem \ref{theorem:parineq}, $V_g\hyperc{\pi}{C_1 \parallel C_1}\geq V_g\hyperc{\pi}{C_1}$ for all priors $\pi$ and gain functions $g$. 

$\bullet$ ($\vchoiceop{p}$) For any gain function $g$ and distribution $\pi \in \dist \calx$
\begin{align*}
&V_g \hyperc{\pi}{C_1 \vchoice{p}C_1}\\
=& pV_g \hyperc{\pi}{C_1}+(1-p)V_g \hyperc{\pi}{C_1}&\text{(from theorem \ref{theorem:reseqvchoice})}\\
=& V_g \hyperc{\pi}{C_1} 
\end{align*}

$\bullet$ ($\hchoiceop{p}$)  We have, for all $x \in \calx$ and $y \in \caly_1$:
$$(C_1 \hchoice{p} C_1)(x,y)= p \cdot C_1(x,y) + (1-p) \cdot C_1(x,y) = C_1(x,y)$$
\qed
\end{proof}

\resdistsame*
\begin{proof}
$\bullet$ ($\parallel$) Using the commutative and associative properties of the parallel operator, it is easy to show that 
$$(C_1 \parallel C_2) \parallel (C_1 \parallel C_3) \equiv (C_1 \parallel C_1) \parallel (C_2 \parallel C_3)$$
Now, from proposition \ref{prop:idem}, $(C_1 \parallel C_1) \refines C_1$.
Thus, theorem \ref{theorem:parord} implies 
$$(C_1 \parallel C_1) \parallel (C_2 \parallel C_3) \refines  C_1 \parallel (C_2 \parallel C_3)$$ 

$\bullet$ ($\vchoiceop{p}$) For all $\pi \in \dist \calx$ and gain functions $g$, we have
\begin{align*}
&V_g \hyperc{\pi}{C_1 \vchoice{p}(C_2 \vchoice{q}C_3)}\\
=&pV_g \hyperc{\pi}{C_1}+(1-p)V_g \hyperc{\pi}{C_2 \vchoice{q}C_3} &\text{(from theorem \ref{theorem:reseqvchoice})}\\
=&pV_g \hyperc{\pi}{C_1}+(1-p)qV_g \hyperc{\pi}{C_2} +(1-p)(1-q)V_g \hyperc{\pi}{C_3} &\text{(from theorem \ref{theorem:reseqvchoice})}\\
=&pqV_g \hyperc{\pi}{C_1}+(1-p)qV_g \hyperc{\pi}{C_2}\\ &+p(1-q)V_g \hyperc{\pi}{C_1}+(1-p)(1-q)V_g \hyperc{\pi}{C_3} &\text{(rearranging)}\\
=&qV_g \hyperc{\pi}{C_1 \vchoice{p} C_2}+(1-q)V_g \hyperc{\pi}{C_1 \vchoice{p} C_3} &\text{(from theorem \ref{theorem:reseqvchoice})}\\
=&V_g \hyperc{\pi}{(C_1 \vchoice{p}C_2) \vchoice{q}(C_1 \vchoice{p}C_3)}&\text{(from theorem \ref{theorem:reseqvchoice})}
\end{align*}

$\bullet$ ($\hchoiceop{p}$) For all $x \in \calx$ and $y \in \caly_1 \cap \caly_2 \cap \caly_3$ ,
\begin{align*}
&(C_1 \hchoice{p} (C_2 \hchoice{q} C_3 ))(x,y)\\
=&pC_1(x, y) +(1-p)qC_2(x,y)+(1-p)(1-q)C_3(x,y) &\text{(by def. of $\hchoice{p}$, $\hchoice{q}$)}\\
=&pqC_1(x, y) +(1-p)qC_2(x,y)\\
&+p(1-q)C_1(x, y)+(1-p)(1-q)C_3(x,y) &\text{(rearranging)}\\
=&q(C_1 \hchoice{p} C_2) (x,y) + (1-q) (C_1 \hchoice{p} C_3) (x, y) &\text{(by def. of $\hchoice{p}$)}\\
=&((C_1 \hchoice{p} C_2) \hchoice{q} (C_1 \hchoice{p} C_3) )(x,y) &\text{(by def. of $\hchoice{q}$)}\\
\end{align*}

For all $x \in \calx$ and $y \in (\caly_1 \cap \caly_2) \setminus \caly_3$ ,
\begin{align*}
&(C_1 \hchoice{p} (C_2 \hchoice{q} C_3 ))(x,y)\\
=&pC_1(x, y) +(1-p)qC_2(x,y) &\text{(by def. of $\hchoice{p}$, $\hchoice{q}$)}\\
=&pqC_1(x, y) +(1-p)qC_2(x,y)+p(1-q)C_1(x, y) &\text{(rearranging)}\\
=&q(C_1 \hchoice{p} C_2) (x,y) + (1-q) (C_1 \hchoice{p} C_3) (x, y) &\text{(by def. of $\hchoice{p}$)}\\
=&((C_1 \hchoice{p} C_2) \hchoice{q} (C_1 \hchoice{p} C_3) )(x,y) &\text{(by def. of $\hchoice{q}$)}\\
\end{align*}

For all $x \in \calx$ and $y \in (\caly_1 \cap \caly_3 )\setminus \caly_2$ ,
\begin{align*}
&(C_1 \hchoice{p} (C_2 \hchoice{q} C_3 ))(x,y)\\
=&pC_1(x, y) +(1-p)(1-q)C_3(x,y) &\text{(by def. of $\hchoice{p}$, $\hchoice{q}$)}\\
=&pqC_1(x, y)+p(1-q)C_1(x, y)+(1-p)(1-q)C_3(x,y) &\text{(rearranging)}\\
=&q(C_1 \hchoice{p} C_2) (x,y) + (1-q) (C_1 \hchoice{p} C_3) (x, y) &\text{(by def. of $\hchoice{p}$)}\\
=&((C_1 \hchoice{p} C_2) \hchoice{q} (C_1 \hchoice{p} C_3) )(x,y) &\text{(by def. of $\hchoice{q}$)}\\
\end{align*}

For all $x \in \calx$ and $y \in (\caly_2 \cap \caly_3 )\setminus \caly_1$ ,
\begin{align*}
&(C_1 \hchoice{p} (C_2 \hchoice{q} C_3 ))(x,y)\\
=&(1-p)qC_2(x,y)+(1-p)(1-q)C_3(x,y) &\text{(by def. of $\hchoice{p}$, $\hchoice{q}$)}\\
=&q(C_1 \hchoice{p} C_2) (x,y) + (1-q) (C_1 \hchoice{p} C_3) (x, y) &\text{(by def. of $\hchoice{p}$)}\\
=&((C_1 \hchoice{p} C_2) \hchoice{q} (C_1 \hchoice{p} C_3) )(x,y) &\text{(by def. of $\hchoice{q}$)}\\
\end{align*}

For all $x \in \calx$ and $y \in \caly_1 \setminus( \caly_2 \cup \caly_3)$ ,
\begin{align*}
&(C_1 \hchoice{p} (C_2 \hchoice{q} C_3 ))(x,y)\\
=&pC_1(x, y) &\text{(by def. of $\hchoice{p}$, $\hchoice{q}$)}\\
=&pqC_1(x, y)+p(1-q)C_1(x, y) &\text{(rearranging)}\\
=&q(C_1 \hchoice{p} C_2) (x,y) + (1-q) (C_1 \hchoice{p} C_3) (x, y) &\text{(by def. of $\hchoice{p}$)}\\
=&((C_1 \hchoice{p} C_2) \hchoice{q} (C_1 \hchoice{p} C_3) )(x,y) &\text{(by def. of $\hchoice{q}$)}\\
\end{align*}

For all $x \in \calx$ and $y \in \caly_2 \setminus ( \caly_1 \cup \caly_3)$ ,
\begin{align*}
&(C_1 \hchoice{p} (C_2 \hchoice{q} C_3 ))(x,y)\\
=&(1-p)qC_2(x,y)&\text{(by def. of $\hchoice{p}$, $\hchoice{q}$)}\\
=&q(C_1 \hchoice{p} C_2) (x,y)  &\text{(by def. of $\hchoice{p}$)}\\
=&((C_1 \hchoice{p} C_2) \hchoice{q} (C_1 \hchoice{p} C_3) )(x,y) &\text{(by def. of $\hchoice{q}$)}\\
\end{align*}

For all $x \in \calx$ and $y \in \caly_3 \setminus (\caly_1 \cup \caly_2)$ ,
\begin{align*}
&(C_1 \hchoice{p} (C_2 \hchoice{q} C_3 ))(x,y)\\
=&(1-p)(1-q)C_3(x,y) &\text{(by def. of $\hchoice{p}$, $\hchoice{q}$)}\\
=& (1-q) (C_1 \hchoice{p} C_3) (x, y) &\text{(by def. of $\hchoice{p}$)}\\
=&((C_1 \hchoice{p} C_2) \hchoice{q} (C_1 \hchoice{p} C_3) )(x,y) &\text{(by def. of $\hchoice{q}$)}\\
\end{align*}
\qed
\end{proof}

\resdistdiff*
\begin{proof}
$\bullet$ ($\parallel$ over $\vchoiceop{p}$) The bijection is given by $\review{\psi}((y_1,(y_2,1)))=((y_1, y_2), 1)$ and $\review{\psi}((y_1,(y_3,2)))=((y_1, y_3), 2)$.
 For all $x \in \calx$, $y_1 \in \caly_1$ and $y_2 \in \caly_2$,
\begin{align*}
&(C_1 \parallel (C_2 \vchoice{p} C_3)) (x, (y_1,(y_2,1))) \\
=&C_1 (x, y_1)(C_2 \vchoice{p} C_3) (x, (y_2,1)) &\text{(by definition of $\parallel$)}\\
=&pC_1 (x, y_1)C_2 (x, y_2) &\text{(by definition of $\vchoice{p}$)}\\
=&p (C_1 \parallel C_2) (x, (y_1, y_2)) &\text{(by definition of $\parallel$)}\\
=&((C_1 \parallel C_2) \vchoice{p} (C_1 \parallel C_3)) (x, ((y_1, y_2), 1)) &\text{(by definition of $\vchoice{p}$)}
\end{align*}

For all $x \in \calx$, $y_1 \in \caly_1$ and $y_3 \in \caly_3$,
\begin{align*}
&(C_1 \parallel (C_2 \vchoice{p} C_3)) (x, (y_1,(y_3,2))) \\
=&C_1 (x, y_1)(C_2 \vchoice{p} C_3) (x, (y_3,2)) &\text{(by definition of $\parallel$)}\\
=&(1-p)C_1 (x, y_1)C_3 (x, y_3) &\text{(by definition of $\vchoice{p}$)}\\
=&(1-p) (C_1 \parallel C_2) (x, (y_1, y_3)) &\text{(by definition of $\parallel$)}\\
=&((C_1 \parallel C_2) \vchoice{p} (C_1 \parallel C_3)) (x, ((y_1, y_3), 2)) &\text{(by definition of $\vchoice{p}$)}
\end{align*}

$\bullet$ ($\parallel$ over $\hchoiceop{p}$) For all $x \in \calx$, $y_1 \in \caly_1$ and $y' \in \caly_2 \cap \caly_3$,
\begin{align*}
&(C_1 \parallel (C_2 \hchoice{p} C_3)) (x,(y_1, y')) \\
=&C_1 (x, y_1)(C_2 \hchoice{p} C_3) (x, y') &\text{(by definition of $\parallel$)}\\
=&C_1 (x, y_1)(pC_2 (x, y')+(1-p)C_3 (x, y')) &\text{(by definition of $\hchoice{p}$)}\\
=&p (C_1 \parallel C_2) (x, (y_1, y'))+(1-p) (C_1 \parallel C_2) (x, (y_1, y')) &\text{(by definition of $\parallel$)}\\
=&((C_1 \parallel C_2) \hchoice{p} (C_1 \parallel C_3) )(x, (y_1, y')) &\text{(by definition of $\hchoice{p}$)}
\end{align*}

For all $x \in \calx$, $y_1 \in \caly_1$ and $y_2 \in \caly_2 \setminus \caly_3$,
\begin{align*}
&(C_1 \parallel (C_2 \hchoice{p} C_3)) (x, (y_1, y_2)) \\
=&C_1 (x, y_1)(C_2 \hchoice{p} C_3) (x, y_2) &\text{(by definition of $\parallel$)}\\
=&C_1 (x, y_1)(pC_2 (x, y_2)) &\text{(by definition of $\hchoice{p}$)}\\
=&p (C_1 \parallel C_2) (x, (y_1, y_2))&\text{(by definition of $\parallel$)}\\
=&((C_1 \parallel C_2) \hchoice{p} (C_1 \parallel C_3) )(x, (y_1, y_2)) &\text{(by definition of $\hchoice{p}$)}
\end{align*}

For all $x \in \calx$, $y_1 \in \caly_1$ and $y_3 \in \caly_3 \setminus \caly_2$,
\begin{align*}
&(C_1 \parallel (C_2 \hchoice{p} C_3)) (x, (y_1, y_3)) \\
=&C_1 (x, y_1)(C_2 \hchoice{p} C_3) (x, y_3) &\text{(by definition of $\parallel$)}\\
=&C_1 (x, y_1)((1-p)C_3 (x, y_3)) &\text{(by definition of $\hchoice{p}$)}\\
=&(1-p) (C_1 \parallel C_2) (x, (y_1, y_3)) &\text{(by definition of $\parallel$)}\\
=&((C_1 \parallel C_2) \hchoice{p} (C_1 \parallel C_3) )(x, (y_1, y_3)) &\text{(by definition of $\hchoice{p}$)}
\end{align*}

$\bullet$ ($\vchoiceop{p}$ over $\hchoiceop{p}$) For all $x \in \calx$ and $y_1 \in \caly_1$,
\begin{align*}
&(C_1 \vchoice{p} (C_2 \hchoice{q} C_3)) (x,(y_1, 1)) \\
=&pC_1 (x,y_1) &\text{(by def. of $\vchoice{p}$)}\\
=&q(pC_1 (x,y_1))+(1-q)(pC_1 (x,y_1)) &\text{(rearranging)}\\
=&q(C_1 \vchoice{p} C_2) (x,(y_1,1))+(1-q)(C_1 \vchoice{p} C_3) (x,(y_1,1)) &\text{(by def. of $\vchoice{p}$)}\\
=&((C_1 \vchoice{p} C_2) \hchoice{q}(C_1 \vchoice{p} C_3)) (x,(y_1,1)) &\text{(by def. of $\hchoice{p}$)}\\
\end{align*}

For all $x \in \calx$ and $y' \in \caly_2 \cap \caly_3$,
\begin{align*}
&(C_1 \vchoice{p} (C_2 \hchoice{q} C_3)) (x,(y', 2)) \\
=&(1-p)(C_2 \hchoice{q} C_3) (x,y') &\text{(by def. of $\vchoice{p}$)}\\
=&(1-p)(qC_2 (x,y')+(1-q)C_3 (x,y')) &\text{(by def. of $\hchoice{p}$)}\\
=&q(1-p) C_2(x,y')+(1-q)(1-p)C_3 (x,y') &\text{(rearranging)}\\
=&q(C_1 \vchoice{p} C_2) (x,(y',2))+(1-q)(C_1 \vchoice{p} C_3) (x,(y',2)) &\text{(by def. of $\vchoice{p}$)}\\
=&((C_1 \vchoice{p} C_2) \hchoice{q}(C_1 \vchoice{p} C_3)) (x,(y',2)) &\text{(by def. of $\hchoice{p}$)}\\
\end{align*}

For all $x \in \calx$ and $y_2 \in \caly_2 \setminus \caly_3$,
\begin{align*}
&(C_1 \vchoice{p} (C_2 \hchoice{q} C_3)) (x,(y_2, 2)) \\
=&(1-p)(C_2 \hchoice{q} C_3) (x,y_2) &\text{(by def. of $\vchoice{p}$)}\\
=&(1-p)(qC_2 (x,y_2)) &\text{(by def. of $\hchoice{p}$)}\\
=&q(1-p) C_2(x,y_2) &\text{(rearranging)}\\
=&q(C_1 \vchoice{p} C_2) (x,(y_2,2)) &\text{(by def. of $\vchoice{p}$)}\\
=&((C_1 \vchoice{p} C_2) \hchoice{q}(C_1 \vchoice{p} C_3)) (x,(y_2,2)) &\text{(by def. of $\hchoice{p}$)}\\
\end{align*}

For all $x \in \calx$ and $y_3 \in \caly_3 \setminus \caly_2$,
\begin{align*}
&(C_1 \vchoice{p} (C_2 \hchoice{q} C_3)) (x,(y_3, 2)) \\
=&(1-p)(C_2 \hchoice{q} C_3) (x,y_3) &\text{(by def. of $\vchoice{p}$)}\\
=&(1-p)((1-q)C_3 (x,y_3)) &\text{(by def. of $\hchoice{p}$)}\\
=&(1-q)(1-p)C_3 (x,y_3) &\text{(rearranging)}\\
=&q(1-q)(C_1 \vchoice{p} C_3) (x,(y_3,2)) &\text{(by def. of $\vchoice{p}$)}\\
=&((C_1 \vchoice{p} C_2) \hchoice{q}(C_1 \vchoice{p} C_3)) (x,(y_3,2)) &\text{(by def. of $\hchoice{p}$)}\\
\end{align*}
\qed 
\end{proof}

\rescascparallel*
\begin{proof}
 For all $x \in \calx$, $z_1 \in \calz_1$ and $z_2 \in \calz_2$,
\begin{align*}
&((C_1 D_1)\parallel (C_2D_2))(x,(z_1, z_2))\\
=&(C_1D_1)(x, z_1)(C_2D_2)(x, z_2) &\text{(by def. of $\parallel$)}\\
=&\left( \sum_{y_1 \in \caly_1} C_1(x,y_1)D_1(y_1,z_1)\right)\left( \sum_{y_2 \in \caly_2} C_2(x,y_2)D_2(y_2,z_2)\right)&\text{(by matrix mult.)}\\
=& \sum_{y_1 \in \caly_1}\sum_{y_2 \in \caly_2} C_1(x,y_1)C_2(x,y_2)D_1(y_1,z_1) D_2(y_2,z_2)&\text{(rearranging)}\\
=& \sum_{y_1 \in \caly_1}\sum_{y_2 \in \caly_2} (C_1 \parallel C_2)(x, (y_1, y_2)) D^{\parallel}((y_1, y_2), (z_1, z_2))&\text{(by def. of $\parallel$,  $D^{\parallel}$)}\\
=&((C_1 \parallel C_2)D^{\parallel})(x,(z_1, z_2)) &\text{(by matrix mult.)}
\end{align*}
\qed
\end{proof}

\rescascvchoice*
\begin{proof}
For all $x \in \calx$ and $z_1 \in \calz_1$, 
\begin{align*}
&((C_1 D_1)\vchoice{p} (C_2D_2))(x,(z_1,1))\\
=&p(C_1D_1)(x, z_1)&\text{(by def. of $\vchoice{p}$)}\\
=&p \sum_{y_1 \in \caly_1} C_1(x,y_1)D_1(y_1,z_1)&\text{(by matrix mult.)}\\
=&\sum_{y_1 \in \caly_1} (pC_1(x,y_1))D_1(y_1,z_1)&\text{(by matrix mult.)}\\
=&\sum_{y_1 \in \caly_1} (C_1 \vchoice{p} C_2)(x,(y_1, 1))D_1(y_1,z_1)&\text{(by def. of $\vchoice{p}$)}\\
=&\sum_{(y,i) \in \caly_1 \sqcup \caly_2} (C_1 \vchoice{p} C_2)(x,(y_i, i))D^{\vchoiceop{}}((y_i,i),(z_1, 1))&\text{(by def. of $D^{\vchoiceop{}}$)}\\
=&((C_1 \vchoice{p} C_2)D^{\vchoiceop{}})(x,(z_1, 1)) &\text{(by matrix mult.)}
\end{align*}
Similarly, for all $x \in \calx$ and $z_2 \in \calz_2$,
\begin{align*}
&((C_1 D_1)\vchoice{p} (C_2D_2))(x,(z_2,2))\\
=&(1-p)(C_2D_2)(x, z_2)&\text{(by def. of $\vchoice{p}$)}\\
=&(1-p) \sum_{y_2 \in \caly_2} C_2(x,y_2)D_2(y_2,z_2)&\text{(by matrix mult.)}\\
=&\sum_{y_2 \in \caly_2} ((1-p)C_2(x,y_2))D_2(y_2,z_2)&\text{(by matrix mult.)}\\
=&\sum_{y_2 \in \caly_2} (C_1 \vchoice{p} C_2)(x,(y_2, 2))D_2(y_2,z_2)&\text{(by def. of $\vchoice{p}$)}\\
=&\sum_{(y,i) \in \caly_1 \sqcup \caly_2} (C_1 \vchoice{p} C_2)(x,(y_i, i))D^{\vchoiceop{}}((y_i,i),(z_2, 2))&\text{(by def. of $D^{\vchoiceop{}}$)}\\
=&((C_1 \vchoice{p} C_2)D^{\vchoiceop{}})(x,(z_2, 2)) &\text{(by matrix mult.)}
\end{align*}
\qed
\end{proof}

\rescaschchoice*

\begin{proof}
For all $x \in \calx$ and $z \in \calz$,
\begin{align*}
&((C_1D)\hchoice{p}(C_2D))(x,z)\\
=&p(C_1D)(x,z)+(1-p)(C_2D)(x,z) &\text{(by def. of $\hchoiceop{p}$)}\\
=&p \sum_{y \in \caly} C_1 (x,y)D(y,z) + (1-p)\sum_{y \in \caly} C_2 (x,y)D(y,z) &\text{(by matrix mult.)} \\
=&\sum_{y \in \caly} (pC_1 (x,y)+(1-p)C_2(x,y))D(y,z)  &\text{(rearranging)} \\
=&\sum_{y \in \caly} (C_1 \hchoice{p} C_2)(x,y)D(y,z)  &\text{(by def. of $\hchoiceop{p}$)} \\
=&((C_1 \hchoice{p} C_2)D)(x,z) &\text{(by matrix mult.)} \\
\end{align*}
\end{proof}

\subsection{Proofs of Section~\ref{sec:leakage}}

\resuplowparallel*
\begin{proof}
$\bullet$ (Lower bound) For all $\pi \in \dist \calx$ and gain functions $g$,
\begin{align*}
&V_g\hyperc{\pi}{C_1 \parallel C_2} \\
 =&\sum_{y_1 \in \caly_1} \sum_{y_2 \in \caly_2} \max\limits_{w \in \mathcal{W}} \sum_{x \in \mathcal{X}} (C_1 \parallel C_2) (x,(y_1,y_2))  g(w,x)  \pi(x) &\text{(by def. of $V_g$)}\\
 =&\sum_{y_1 \in \caly_1} \sum_{y_2 \in \caly_2} \max\limits_{w \in \mathcal{W}} \sum_{x \in \mathcal{X}} C_1 (x,y_1)  C_2 (x,y_2)  g(w,x)  \pi(x) &\text{(by def. of $\parallel$)}\\
 \geq&\sum_{y_1 \in \caly_1} \max\limits_{w \in \mathcal{W}} \sum_{y_2 \in \caly_2} \sum_{x \in \mathcal{X}} C_1 (x,y_1)  C_2 (x,y_2)  g(w,x)  \pi(x) &\text{(taking \qm{max} out)}\\
 =&\sum_{y_1 \in \caly_1} \max\limits_{w \in \mathcal{W}} \sum_{x \in \mathcal{X}} C_1 (x,y_1)  g(w,x)  \pi(x)  \sum_{y_2 \in \caly_2}  C_2 (x,y_2) &\text{(rearranging)}\\
 =&\sum_{y_1 \in \caly_1} \max\limits_{w \in \mathcal{W}} \sum_{x \in \mathcal{X}} C_1 (x,y_1)  g(w,x)  \pi(x) &\text{($C_2$ is a channel)}\\
 =&V_g\hyperc{\pi}{C_1} &\text{(by def. of $V_g$)}
\end{align*}
The proof that $V_g\hyperc{\pi}{C_1 \parallel C_2} \geq V_g\hyperc{\pi}{C_2}$ is analogous.\\

$\bullet$ (Upper bound) For all $\pi \in \dist \calx$, and gain functions $g$,
\begin{align*}
&V_g\hyperc{\pi}{C_1 \parallel C_2} \\
 =&\sum_{y_2 \in \caly_2} \sum_{y_1 \in \caly_1} \max\limits_{w \in \mathcal{W}} \sum_{x \in \mathcal{X}} (C_1 \parallel C_2) (x,(y_1,y_2))  g(w,x)  \pi(x) &\text{(by def. of $V_g$)}\\
 =&\sum_{y_2 \in \caly_2} \sum_{y_1 \in \caly_1} \max\limits_{w \in \mathcal{W}} \sum_{x \in \mathcal{X}} C_1 (x,y_1)  C_2 (x,y_2)  g(w,x)  \pi(x) &\text{(by def. of $\parallel$)}\\
  =&\sum_{y_2 \in \caly_2} \sum_{y_1 \in \caly_1} \max\limits_{w \in \mathcal{W}} \sum_{x \in \calx'} C_1 (x,y_1)  C_2 (x,y_2)  g(w,x)  \pi(x) \\&\text{(if $x \not\in \calx'$,$g(w,x)  \pi(x)=0$  )}\\
  \leq&\sum_{y_2 \in \caly_2} \sum_{y_1 \in \caly_1} \max\limits_{w \in \mathcal{W}} \sum_{x \in \mathcal{X}'} C_1 (x,y_1)\left(\max\limits_{x' \in \calx'} C_2(x',y_2) \right)  g(w,x)  \pi(x)\\ &\text{( for all $x \in \calx'$, $C_2(x,y_2) \leq \max\limits_{x' \in \calx'} C_2(x',y_2)$)}\\
   =&\sum_{y_2 \in \caly_2} \max\limits_{x' \in \calx'} C_2(x',y_2)  \sum_{y_1 \in \caly_1} \max\limits_{w \in \mathcal{W}} \sum_{x \in \mathcal{X}'} C_1 (x,y_1) g(w,x)  \pi(x) &\text{(rearranging)}\\
  =&\sum_{y_2 \in \caly_2} \max\limits_{x \in \calx'} C_2(x,y_2)   V_g\hyperc{\pi}{C_1} &\text{(by def. of $V_g$)}\\
 =&V_g\hyperc{\pi}{C_1}  \sum_{y_2 \in \caly_2} \max\limits_{x \in \calx'} C_2(x,y_2)  
\end{align*}
The proof that $V_g\hyperc{\pi}{C_1 \parallel C_2} \leq V_g \hyperc{\pi}{C_2}  \sum_{y_1 \in \caly_1} \max\limits_{x \in \calx'} C_1(x,y_1) $ is analogous.
\qed
\end{proof}

\reseqvchoice*
\begin{proof}
For all $\pi \in \dist \calx$, gain functons $g$ and $p \in [0,1]$,
\begin{align*}
 &V_g\hyperc{\pi}{C_1 \vchoice{p} C_2} \\
 =&\sum_{(y,i) \in \caly_1 \sqcup \caly_2} \max\limits_{w \in \mathcal{W}} \sum_{x \in \mathcal{X}} (C_1 \vchoice{p} C_2) (x,(y,i))  g(w,x)  \pi(x) &\text{(by def. of $V_g$)}\\
 =&\sum_{y \in \caly_1}\max\limits_{w \in \mathcal{W}} \sum_{x \in \mathcal{X}} (C_1 \vchoice{p} C_2) (x,(y,1))  g(w,x)  \pi(x)\\
 &+\sum_{y \in \mathcal{Y}_2}\max\limits_{w \in \mathcal{W}} \sum_{x \in \mathcal{X}} (C_1 \vchoice{p} C_2) (x,(y,2))  g(w,x)  \pi(x) &\text{(separating $\caly_1$ and $\caly_2$)}\\
 =&\sum_{y \in \mathcal{Y}_1}\max\limits_{w \in \mathcal{W}} \sum_{x \in \mathcal{X}} p  C_1 (x,y) g(w,x)  \pi(x) \\
 &+\sum_{y \in \mathcal{Y}_2}\max\limits_{w \in \mathcal{W}} \sum_{x \in \mathcal{X}} (1-p)  C_2 (x,y) g(w,x)  \pi(x) &\text{(by def. of $\vchoiceop{p}$)}\\
 =&p  V_g\hyperc{\pi}{C_1}+(1-p)  V_g\hyperc{\pi}{C_2} &\text{(by def. of $V_g$)}\\
\end{align*}
\qed
\end{proof}

\resuplowhchoice*
\begin{proof}
$\bullet$ (Lower bound) For all $\pi \in \calx$, gain functions $g$ and $p \in [0,1]$,
\begin{align*}
 &V_g\hyperc{\pi}{C_1 \hchoice{p} C_2} \\
 =&\sum_{y \in \caly_1 \cup \caly_2} \max\limits_{w \in \mathcal{W}} \sum_{x \in \mathcal{X}} (C_1 \hchoice{p} C_2) (x,y)  g(w,x)  \pi(x) &\text{(by def. of $V_g$)}\\
 =&\sum_{y \in \caly_1}\max\limits_{w \in \mathcal{W}} \sum_{x \in \mathcal{X}} (C_1 \hchoice{p} C_2) (x,y)  g(w,x)  \pi(x)\\
 &+\sum_{y \in \caly_2 \setminus \caly_1}\max\limits_{w \in \mathcal{W}} \sum_{x \in \calx}(C_1 \hchoice{p} C_2) (x,y)  g(w,x)  \pi(x) &\text{(rearranging)}\\
 \geq&\sum_{y \in \caly_1}\max\limits_{w \in \mathcal{W}} \sum_{x \in \mathcal{X}} (C_1 \hchoice{p} C_2) (x,y)  g(w,x)  \pi(x) &\text{(each sum is nonegative)}\\
  \geq&\sum_{y \in \caly_1}\max\limits_{w \in \mathcal{W}} \sum_{x \in \mathcal{X}} p  C_1 (x,y)  g(w,x)  \pi(x) &\text{(by def. of $\hchoiceop{p}$)}\\
 =&p  V_g\hyperc{\pi}{C_1} &\text{(by def. of $V_g$)}
\end{align*}
The proof that $(1-p)  V_g\hyperc{\pi}{C_2} \leq V_g\hyperc{\pi}{C_1 \hchoice{p} C_2}$ is similar.

$\bullet$ (Upper bound) For all $\pi \in \dist \calx$, gain functons $g$ and $p \in [0,1]$,
\begingroup
\allowdisplaybreaks
\begin{align*}
 &V_g\hyperc{\pi}{C_1 \hchoice{p} C_2} \\
 =&\sum_{y \in \caly_1 \cup \caly_2} \max\limits_{w \in \mathcal{W}} \sum_{x \in \mathcal{X}} (C_1 \hchoice{p} C_2) (x,y)  g(w,x)  \pi(x) &\text{(by def. of $V_g$)}\\
 =&\sum_{y \in \caly_1 \cap \caly_2}\max\limits_{w \in \mathcal{W}} \sum_{x \in \mathcal{X}} (C_1 \hchoice{p} C_2) (x,y)  g(w,x)  \pi(x)\\
 &+\sum_{y \in \mathcal{Y}_1 \setminus \caly_2}\max\limits_{w \in \mathcal{W}} \sum_{x \in \mathcal{X}} (C_1 \hchoice{p} C_2) (x,y)  g(w,x)  \pi(x)\\
 &+\sum_{y \in \mathcal{Y}_2 \setminus \caly_1}\max\limits_{w \in \mathcal{W}} \sum_{x \in \mathcal{X}} (C_1 \hchoice{p} C_2) (x,y)  g(w,x)  \pi(x)
 &\text{(rearranging)}\\
 =&\sum_{y \in \caly_1 \cap \caly_2}\max\limits_{w \in \mathcal{W}} \sum_{x \in \mathcal{X}} (pC_1(x,y)+ (1-p)C_2 (x,y))  g(w,x)  \pi(x)\\
 &+\sum_{y \in \mathcal{Y}_1 \setminus \caly_2}\max\limits_{w \in \mathcal{W}} \sum_{x \in \mathcal{X}} pC_1(x,y)  g(w,x)  \pi(x)\\
 &+\sum_{y \in \mathcal{Y}_2 \setminus \caly_1}\max\limits_{w \in \mathcal{W}} \sum_{x \in \mathcal{X}} (1-p)C_2 (x,y)  g(w,x)  \pi(x)
 &\text{(by def. of $\hchoiceop{p}$)}\\
 \leq&\sum_{y \in \caly_1 \cap \caly_2}\max\limits_{w \in \mathcal{W}} \sum_{x \in \mathcal{X}} pC_1(x,y)  g(w,x)  \pi(x)\\
 &+\sum_{y \in \caly_1 \cap \caly_2}\max\limits_{w \in \mathcal{W}} \sum_{x \in \mathcal{X}}(1-p)C_2 (x,y)   g(w,x)  \pi(x)\\
 &+\sum_{y \in \mathcal{Y}_1 \setminus \caly_2}\max\limits_{w \in \mathcal{W}} \sum_{x \in \mathcal{X}} pC_1(x,y)  g(w,x)  \pi(x)\\
 &+\sum_{y \in \mathcal{Y}_2 \setminus \caly_1}\max\limits_{w \in \mathcal{W}} \sum_{x \in \mathcal{X}} (1-p)C_2 (x,y)  g(w,x)  \pi(x)
 &\text{(distributing)}\\
 =&\sum_{y \in \mathcal{Y}_1}\max\limits_{w \in \mathcal{W}} \sum_{x \in \mathcal{X}} pC_1(x,y)  g(w,x)  \pi(x)\\
 &+\sum_{y \in \mathcal{Y}_2}\max\limits_{w \in \mathcal{W}} \sum_{x \in \mathcal{X}} (1-p)C_2 (x,y)  g(w,x)  \pi(x)
 &\text{(rearranging)}\\
 =&p  V_g\hyperc{\pi}{C_1}+(1-p)  V_g\hyperc{\pi}{C_2} &\text{(by def. of $V_g$)}
\end{align*}
\qed
\endgroup
\end{proof}

\resoperorder*

\begin{proof}
This proof relies on theorems \ref{theorem:parineq}, \ref{theorem:reseqvchoice} and \ref{theorem:hineq}.
 The fact that $ V_g\hyperc{\pi}{C_1 \vchoice{p} C_2}\geq V_g\hyperc{\pi}{C_1 \hchoice{p} C_2}$ is immediate from theorems  \ref{theorem:reseqvchoice} and \ref{theorem:hineq}. 
To see that $V_g\hyperc{\pi}{C_1 \parallel C_2} \geq V_g\hyperc{\pi}{C_1 \vchoice{p} C_2}$, notice that, for any gain function $g$ and prior $\pi \in \dist \calx$,
\begin{align*}
&V_g\hyperc{\pi}{C_1 \parallel C_2}\\
\geq& \max (V_g \hyperc{\pi}{C_1} ,V_g \hyperc{\pi}{C_2}) &\text{(by theorem \ref{theorem:parineq})}\\
\geq& p  V_g\hyperc{\pi}{C_1}+(1-p)  V_g\hyperc{\pi}{C_2} &\text{($V_g\hyperc{\pi}{C_i} \leq \max (V_g \hyperc{\pi}{C_1} ,V_g \hyperc{\pi}{C_2})$)}\\
=&V_g\hyperc{\pi}{C_1 \vchoice{p} C_2}&\text{(by theorem \ref{theorem:reseqvchoice})} 
\end{align*}
\qed
\end{proof}

\ressecvchoice*
\begin{proof}
We have, for all $p\in (0,1]$, $\pi \in \dist \calx$ and gain functions $g$ 
\begin{align*}
     &V_g\hyperc{\pi}{C_1}\leq V_g\hyperc{\pi}{C_2}\\
     \Leftrightarrow &p  V_g\hyperc{\pi}{C_1}\leq p  V_g\hyperc{\pi}{C_2} &\text{($p>0$)}\\
     \Leftrightarrow &\forall C.\quad p  V_g\hyperc{\pi}{C_1}+ (1-p)  V_g\hyperc{\pi}{C}\\
     &\leq p  V_g\hyperc{\pi}{C_2}+ (1-p)  V_g\hyperc{\pi}{C}&\text{(adding in both sides)}\\
     \Leftrightarrow &\forall C.\quad V_g\hyperc{\pi}{C_1 \vchoice{p} C} \leq V_g\hyperc{\pi}{C_2 \vchoice{p} C}  &\text{(from theorem \ref{theorem:reseqvchoice})}
\end{align*}
\qed
\end{proof}

\ressecparallel*
\begin{proof}

$\bullet$ (First equivalence) $(\Rightarrow)$ We will prove the contrapositive.

Let $g$ be a gain function. Assume that there is a probability distribution $\pi \in \dist \calx$ and a compatible channel $C: \calx \times \mathcal{Z}\rightarrow \mathbb{R}$ to
$C_1$ and $C_2$, such that
$$V_g\hyperc{\pi}{C_1 \parallel C} > V_g\hyperc{\pi}{C_2 \parallel C}.$$

This means, by definition, that
\begin{multline*}
\sum_{z \in \mathcal{Z}}\sum_{y_1 \in \mathcal{Y}_1} \max\limits_{w \in \mathcal{W}} \bigg(\sum_{x \in \mathcal{X}} C_1(x,y_1) C(x,z)  g(w,x)  \pi(x)  \bigg) > \\
\sum_{z \in \mathcal{Z}}\sum_{y_2 \in \mathcal{Y}_2} \max\limits_{w \in \mathcal{W}}\bigg(\sum_{x \in \mathcal{X}} C_2(x,y_2) C(x,z) g(w,x)   \pi(x)  \bigg).
\end{multline*}

For this inequality to hold, it must be true that, for some $z' \in \mathcal{Z}$,

\begin{multline*}
\sum_{y_1 \in \mathcal{Y}_1} \max\limits_{w \in \mathcal{W}} \bigg(\sum_{x \in \mathcal{X}} C_1(x,y_1) C(x,z')  g(w,x)  \pi(x)  \bigg) > \\
\sum_{y_2 \in \mathcal{Y}_2} \max\limits_{w \in \mathcal{W}}\bigg(\sum_{x \in \mathcal{X}} C_2(x,y_2) C(x,z') g(w,x)   \pi(x)  \bigg).
\end{multline*}

The above inequality also implies that $\sum_{x' \in \mathcal{X}}(  C(x',z')   \pi(x')  )>0$, otherwise the LHS couldn't possibly be strictly greater than the RHS. We can therefore divide both sides by this quantity. Being a positive constant, we can put it \qm{inside} the max in both sides, yielding:

\begin{align*}
\sum_{y_1 \in \mathcal{Y}_1} \max\limits_{w \in \mathcal{W}}\sum_{x \in \calx}\bigg(  C_1(x,y_1)  g(w,x) \frac{ C(x,z') \pi(x)}{\sum_{x' \in \mathcal{X}}(  C(x',z')   \pi(x')  )}  \bigg) > \\
\sum_{y_1 \in \mathcal{Y}_2} \max\limits_{w \in \mathcal{W}}\sum_{x \in \calx} \bigg( C_2(x,y_2)  g(w,x) \frac{C(x,z') \pi(x)}{\sum_{x' \in \mathcal{X}}(  C(x',z')   \pi(x')  )}  \bigg)
\end{align*}

We now define $\pi': \mathcal{X} \rightarrow \mathbb{R}$ by 
$$\pi '(x) =\frac{C(x,z') \pi(x)}{\sum_{x' \in \mathcal{X}}(  C(x',z')   \pi(x')  )}$$

It is clear that $\pi ' \in \mathbb{D} \mathcal{X}$, for it is a non-negative function whose values sum to 1. Therefore, the above inequality reduces to
$$\sum_{y_1 \in \mathcal{Y}_1} \max\limits_{w \in \mathcal{W}} \sum_{x \in \mathcal{X}} C_1(x,y_1) g(w,x)  \pi'(x) >\sum_{y_2 \in \mathcal{Y}_2} \max\limits_{w \in \mathcal{W}} \sum_{x \in \mathcal{X}} C_2(x,y_2)  g(w,x) \pi'(x)  .$$

That is, $V_g\hyperc{\pi'}{C_1} > V_g \hyperc{\pi'}{C_2}$, which completes the proof.

$(\Leftarrow)$ Assuming that the right hand side holds, and recalling proposition \ref{prop:null},  we have, for all $\pi \in \mathbb{D}\mathcal{X}$
$$V_{g}\hyperc{\pi}{C_1}=V_{g}\hyperc{\pi}{C_1 \parallel \nullchannel} \leq V_{g}\hyperc{\pi}{C_2 \parallel \nullchannel}=V_{g}\hyperc{\pi}{C_2}.$$\\

$\bullet$ (Second equivalence) The proof is very similar to the one above

$(\Rightarrow)$ We will, again, prove the contrapositive.

Let $\pi$ be a prior distribution. Assume that there is a gain function $g$ and a compatible channel $C: \calx \times \mathcal{Z}\rightarrow \mathbb{R}$ to
$C_1$ and $C_2$, such that
$$V_g\hyperc{\pi}{C_1 \parallel C} > V_g\hyperc{\pi}{C_2 \parallel C}.$$

Similarly to the proof above, this implies that $\exists z' \in \mathcal{Z}$ such that

\begin{multline*}
\sum_{y_1 \in \mathcal{Y}_1} \max\limits_{w \in \mathcal{W}} \bigg(\sum_{x \in \mathcal{X}} C_1(x,y_1) C(x,z')  g(w,x)  \pi(x)  \bigg) > \\
\sum_{y_2 \in \mathcal{Y}_2} \max\limits_{w \in \mathcal{W}}\bigg(\sum_{x \in \mathcal{X}} C_2(x,y_2) C(x,z') g(w,x)   \pi(x)  \bigg).
\end{multline*}

We now define another gain function $g'(w,x)=C(x,z') g(w,x)$. Substituting this value in the inequality above, we get

$$\sum_{y_1 \in \mathcal{Y}_1} \max\limits_{w \in \mathcal{W}} \sum_{x \in \mathcal{X}} C_1(x,y_1) g' (w,x)  \pi(x) >
\sum_{y_2 \in \mathcal{Y}_2} \max\limits_{w \in \mathcal{W}} \sum_{x \in \mathcal{X}} C_2(x,y_2)  g'(w,x) \pi(x) .$$

That is, $V_{g'}\hyperc{\pi}{C_1} > V_{g'} \hyperc{\pi}{C_2}$, which completes the proof.
\newline

$(\Leftarrow)$ Assuming that the right hand side of the equivalence holds for a given $\pi$, and recalling proposition \ref{prop:null},  we have, for all gain functions $g$:
$$V_{g}\hyperc{\pi}{C_1}=V_{g}\hyperc{\pi}{C_1 \parallel \nullchannel} \leq V_{g}\hyperc{\pi}{C_2 \parallel \nullchannel}=V_{g}\hyperc{\pi}{C_2}.$$
\qed
\end{proof}

\ressechchoicei*
\begin{proof}
 Let $\calx=\{1,2\}$, $\pi_u = (0.5,0.5)$ $g_{id}: \calx \times \calx \rightarrow [0,1]$ 
such that $g_{id}(x,x')=1$ if $x=x'$ and $0$ otherwise. We divide the proof in two cases:

\paragraph{Case 1: $p \leq 0.5$}
Take
\[
C_1=
\begin{bmatrix}
    0.5       & 0.5   \\
    0.5       & 0.5   \\
\end{bmatrix}
, \quad
C_2=
\begin{bmatrix}
    1       & 0   \\
    0       & 1   \\
\end{bmatrix}
, \quad
C=
\begin{bmatrix}
    0       & 1   \\
    1       & 0   \\
\end{bmatrix}
.
\]

Notice that 
$C_1$ is a null channel, and $C_2$ is a transparent channel. Thus, $\forall \pi \in \dist \calx, \forall g. \; V_{g}\hyperc{\pi}{C_1}  \leq  V_{g}\hyperc{\pi}{C_2}$.

However, we have
\[
C_1  \hchoice{p}  C=
\begin{bmatrix}
    \nicefrac{p}{2}       & 1-\nicefrac{p}{2}   \\
    1-\nicefrac{p}{2}       & \nicefrac{p}{2}   \\
\end{bmatrix}
, \quad
C_2  \hchoice{p}  C=
\begin{bmatrix}
    p       & 1-p   \\
    1-p       & p   \\
\end{bmatrix}
.
\]

Then, $V_{g_{id}}\hyperc{\pi_u}{C_1  \hchoice{p}  C} = 1- \nicefrac{p}{2}$ and $V_{g_{id}}\hyperc{\pi_u}{C_2  \hchoice{p}  C} = 1-p$. Therefore, $V_{g_{id}}\hyperc{\pi_u}{C_1  \hchoice{p}  C}>V_{g_{id}}\hyperc{\pi_u}{C_2  \hchoice{p}  C}$.

\paragraph{Case 2: $p >0.5$}

Take
\[
C_1=
\begin{bmatrix}
    \nicefrac{1}{2p}-\nicefrac{1}{2}  \;&\; \nicefrac{3}{2} - \nicefrac{1}{2p}   \\
    \nicefrac{3}{2} - \nicefrac{1}{2p} \;&\; \nicefrac{1}{2p}-\nicefrac{1}{2}  \\
\end{bmatrix}
, \quad
C_2=
\begin{bmatrix}
    1       & 0   \\
    0       & 1   \\
\end{bmatrix}
, \quad
C=
\begin{bmatrix}
    0       & 1   \\
    1       & 0   \\
\end{bmatrix}
.
\]

Because $C_2$ is a transparent channel, $\forall \pi \in \dist \calx, \forall g. \; V_{g}\hyperc{\pi}{C_1}  \leq  V_{g}\hyperc{\pi}{C_2}$.

We have
\[
C_1  \hchoice{p}  C=
\begin{bmatrix}
    \nicefrac{1}{2}-\nicefrac{p}{2} \;&\; \nicefrac{p}{2}+\nicefrac{1}{2}   \\
    \nicefrac{p}{2}+\nicefrac{1}{2} \;&\; \nicefrac{1}{2}-\nicefrac{p}{2}     \\
\end{bmatrix}
, \quad
C_2  \hchoice{p}  C=
\begin{bmatrix}
    p       & 1-p   \\
    1-p       & p   \\
\end{bmatrix}
.
\]

Thus, $V_{g_{id}}\hyperc{\pi_u}{C_1 \hchoice{p}  C} = \frac{1}{2}(p+1)$ and $V_{g_{id}}\hyperc{\pi_u}{C_2  \hchoice{p}  C} = p$. Hence $V_{g_{id}}\hyperc{\pi_u}{C_1  \hchoice{p}  C}>V_{g_{id}}  \hyperc{\pi_u}{C_2  \hchoice{p}  C}$.

\qed
\end{proof}
\ressechchoiceii*

\begin{proof}
For all $\pi \in \dist \calx$ and gain functions $g$,
\begin{align*}
    &V_g\hyperc{\pi}{C_1}\\
    =&V_g\hyperc{\pi}{C_1 \hchoice{p} C_1} &\text{(from proposition \ref{prop:idem})}\\
    \leq &V_g\hyperc{\pi}{C_2 \hchoice{p} C_1} &\text{(from assumption)}\\
    \leq &p  V_g\hyperc{\pi}{C_2}+(1-p)  V_g\hyperc{\pi}{C_1} &\text{(from theorem \ref{theorem:hineq})}
\end{align*}
Therefore, $V_g\hyperc{\pi}{C_1}\leq p  V_g\hyperc{\pi}{C_2}+(1-p)  V_g\hyperc{\pi}{C_1}$, which yields $V_g\hyperc{\pi}{C_1}\leq  V_g\hyperc{\pi}{C_2}$
\qed
\end{proof}

\subsection{Proofs of Section~\ref{sec:casestudy}}
In this section, we prove theorem \ref{theorem:bound}. First, we need some auxiliary results.

\begin{lemma} \label{lemma:crassoc}
Let $n$ be a positive integer and $\{A_i\}_{i \in \{0,1,...,n\}}$ a collection of compatible channels. If $q$, $p_{} \in [0,1]$ and are not both $0$, then
\begin{align*}
A_0 \hchoice{q}(A_1 \hchoice{p}(...\hchoice{p}
A_{n})&=((...(A_0 \hchoice{\frac{t_0}{t_1}} A_1)\hchoice{\frac{t_1}{t_2}} ...)\hchoice{t_{(n-1)}} A_{n}),\mbox{ for even $n$}\\
A_0 \hchoice{q}(A_1 \hchoice{p}(...\hchoice{q}
A_{n})&=((...(A_0 \hchoice{\frac{t_0}{t_1}} A_1)\hchoice{\frac{t_1}{t_2}} ...)\hchoice{t_{(n-1)}} A_{n}),\mbox{ for odd $n$}
\end{align*}
where 
\begin{align}
t_{2i}&=1-(1-q)^{i+1}(1-p_{})^i, \mbox{ and } \label{eq:tneven} \\
t_{(2i+1)}&=1-(1-q)^{i+1}(1-p_{})^{i+1}. \label{eq:tnodd}
\end{align}

\end{lemma}

\begin{proof}
Firstly,  from proposition \ref{prop:associative} we deduce the following equivalence, for any probabilities $u,v \in [0,1]$ (with $u$ and $v$ not both $0$) and compatible channels $C_1$, $C_2$ and $C_3$:
\begin{equation} \label{eq:assochid}
C_1 \hchoice{u}(C_2 \hchoice{v} C_3)=(C_1\hchoice{\frac{u}{u+v-uv}}C_2)\hchoice{(u+v-uv)} C_3
\end{equation}

Now, we will proceed to the proof. We will prove by induction on $n$ on the set of  positive integers.

The case when $n=1$ is immediate . Let us supposed it is proven for $n\geq1$. Then, if $n$ is odd,
\begin{align*}
&A_0 \hchoice{q}(...\hchoice{p}(A_{n-1}\hchoice{q} (A_{n} \hchoice{p}A_{n+1}))) \\
=&A_0 \hchoice{q}(...\hchoice{p}(A_{n-1}\hchoice{q} A_{n}')...)) &\text{(let $A_{n}'{=}A_{n}{ \hchoice{p}}A_{n+1} 
$)} \\
=&((...(A_0 \hchoice{\frac{t_0}{t_1}}  ...)\hchoice{\frac{t_{(n-2)}}{t_{(n-1)}}}A_{n-1})\hchoice{t_{(n-1)}} A_n' &\text{(by ind. hyp.)}\\
=&((...(A_0 \hchoice{\frac{t_0}{t_1}}  ...)\hchoice{\frac{t_{(n-2)}}{t_{(n-1)}}}A_{n-1})\hchoice{t_{(n-1)}} (A_{n} \hchoice{p}A_{n+1}) &\text{($A_{n}'=A_{n} \hchoice{p}A_{n+1}$)}\\
=&((...(A_0 \hchoice{\frac{t_0}{t_1}}  ...)\hchoice{\frac{t_{(n-1)}}{t_n}} A_{n}) \hchoice{t_n}A_{n+1}  &\text{(by eq. \eqref{eq:assochid} and \eqref{eq:tnodd})}
\end{align*}
The proof for when $n$ is even is almost identical.
\qed
\end{proof}

\begin{lemma} \label{lemma:crlimit}
If $q$ and $p_{}$ are not both $0$, $\lim\limits_{i \rightarrow \infty} C_i$ exists.
\end{lemma}
\begin{proof}
Each $C_i$ is a channel of type $\mathcal{U} \times (\mathcal{D} \cup \mathcal{S})$, and can thus be understood as an element of the set $\reals^{|\mathcal{U} \times (\mathcal{D} \cup \mathcal{S})|}$. 
Thus, if each entry of $C_i$ converges to a real value as $i \rightarrow \infty$, then $\lim\limits_{i \rightarrow \infty} C_i$ exists.

Having that in mind, we prove that, for all $j,k \in \{1,2,...,n\}$, $\{C_i(u_j,d_k)\}_{i \in \mathbb{N}^*}$ and 
$\{C_i(u_j,s_k)\}_{i \in \mathbb{N}^*}$ are Cauchy sequences in the reals.
We start by proving that $\{C_i(u_j,d_k)\}_{i \in \mathbb{N}^*}$ is a Cauchy sequence.

Let $\epsilon > 0$.
From equations \eqref{eq:tneven} and \eqref{eq:tnodd}, $\lim\limits_{i \rightarrow \infty} t_i=1$. Therefore,  
$\exists M \in \mathbb{N} \setminus \{0\}$ such that $i>M \implies 1-t_i< \nicefrac{\epsilon}{2}$.

Suppose $m_1,m_2 >M+1$. 
By lemma \ref{lemma:crassoc}, we have 
\begin{align*}
C_{m_1}=&(...(I_d \hchoice{\frac{t_0}{t_1}} I_sP_s)\hchoice{\frac{t_1}{t_2}} I_dP_d)\hchoice{\frac{t_2}{t_3}}...)\hchoice{\frac{t_{(2M-1)}}{t_{2M}}}I_dP_d^{M})\hchoice{t_{2M}} D_1\\
C_{m_2}=&(...(I_d \hchoice{\frac{t_0}{t_1}} I_sP_s)\hchoice{\frac{t_1}{t_2}} I_dP_d)\hchoice{\frac{t_2}{t_3}}...)\hchoice{\frac{t_{(2M-1)}}{t_{2M}}}I_dP_d^{M})\hchoice{t_{2M}} D_2
\end{align*}
where $D_i=I_sP_s^{M+1} \hchoice{p}( I_dP_d^{M+1}\hchoice{q}(...\hchoice{q} I_sP_s^{m_i})...)$ 
, for $i \in \{1,2\}$. 
The definition of hidden choice then gives us

\begin{align*}
C_{{m_1}} (u_j, d_k)=&t_{2M}((...(I_d \hchoice{\frac{t_0}{t_1}} ...)\hchoice{\frac{t_{(2M-1)}}{t_{2M}}}I_dP_d^{M})(u_j, d_k)+ (1-t_{2M}) D_1 (u_j, d_k)\\
C_{{m_2}} (u_j, d_k)=&t_{2M}((...(I_d \hchoice{\frac{t_0}{t_1}} ...)\hchoice{\frac{t_{(2M-1)}}{t_{2M}}}I_dP_d^{M})(u_j, d_k)+ (1-t_{2M}) D_2 (u_j, d_k)\\
\end{align*}
Thus, 
\begin{align*}
&|C_{{m_1}} (u_j, d_k) - C_{{m_2}} (u_j, d_k)|\\
=&|(1-t_{2M}) D_1 (u_j, d_k) - (1-t_{2M}) D_2 (u_j, d_k)| \\
\leq&|(1-t_{2M})D_1 (u_j, d_k)|+|(1-t_{2M})D_2 (u_j, d_k)|  &\text{($|a-b| \leq |a|+|b|$)}\\
\leq&|(1-t_{2M})|+|(1-t_{2M})|  &\text{($D_i (u_j, d_k) \leq 1$)}\\
<&\nicefrac{\epsilon}{2}+\nicefrac{\epsilon}{2} = \epsilon &\text{($2M>M$)} 
\end{align*}
The demonstration for $\{C_i (u_j,s_k)\}_{i \in \mathbb{N}}$ is almost identical.

Thus, we have established that for all $j,k \in \{1,2,...,n_c\}$, $\{C_i(u_j,d_k)\}_{i \in \mathbb{N}}$ and $\{C_i(u_j,s_k)\}_{i \in \mathbb{N}}$ are Cauchy sequences,
establishing the existence of $\lim\limits_{i \rightarrow \infty}C_i$
\qed
\end{proof}

\begin{lemma}\label{lemma:overpower}
Let $\pi \in \dist \calx$ for some finite set $\calx$ and $g$ be any gain function. Let $n \in \mathbb{N} \setminus \{0\}$ and $\{A_i\}_{i \in \{1, 2, ..., n\}} $ be a collection of  channels with input $\calx$ such that $$i<j \implies V_g\hyperc{\pi}{A_i}\geq V_g\hyperc{\pi}{A_j}$$

Let $\{p_i\}_{i \in \{1, ..., n-1\}}$ be a collection of real numbers in the interval $[0,1]$. Then, for any $n \in \mathbb{N}$,

$$V_g \hyperc{ \pi}{A_1 \hchoice{p_1}(A_2 \hchoice{p_2}(...\hchoice{p_{n-1}} (A_n)))} \leq V_g\hyperc{\pi}{A_1}$$
\end{lemma}

\begin{proof}
We will prove by induction on the size of the collection . The result is obvious if $n=1$ 
Suppose it is true for $n \geq 1$ and let $\{A_i\}_{i \in \{1,...,n+1\}}$ be a collection of $n+1$ channels  with the property described in the Lemma. 
Then,

\begin{align*}
&V_g \hyperc{ \pi}{A_1 \hchoice{p_1}(A_2 \hchoice{p_2}(...\hchoice{p_{n}}(A_{n+1})))}\\
\leq& p_1 V_g \hyperc{ \pi}{A_1}+(1-p_1) V_g \hyperc{ \pi}{A_2 \hchoice{p_2}(...\hchoice{p_{n}}A_{n+1})} &\text{(by theorem \ref{theorem:hineq})}\\
\leq& p_1 V_g \hyperc{ \pi}{A_1}+(1-p_1) V_g \hyperc{ \pi}{A_2} &\text{(by the ind. hypothesis)}\\
\leq& p_1 V_g \hyperc{ \pi}{A_1}+(1-p_1) V_g \hyperc{ \pi}{A_1}=V_g \hyperc{ \pi}{A_1} &\text{($V_g \hyperc{\pi}{A_1}\geq V_g \hyperc{\pi}{A_2}$)}\\
\end{align*}
\qed
\end{proof}

\rescrowdsbounds*

\begin{proof}
From lemma \ref{lemma:crassoc}, we note that, for any $m'>m$, $C_{m'}$ can be written as
\begin{align}
C_{m'}=&((...(I_d \hchoice{\frac{t_0}{t_1}} I_sP_s)\hchoice{\frac{t_1}{t_2}} I_dP_d)\hchoice{\frac{t_2}{t_3}}...)\hchoice{\frac{t_{2m-1}}{t_{2m}}}I_dP_d^{m})\hchoice{t_{2m}} D \nonumber ,\\
C_{m'}=&K_m\hchoice{t_{2m}} D. \label{eq:decomposition}
\end{align}

Where $D=I_sP_s^{m+1} \hchoice{p}( I_dP_d^{m+1}\hchoice{q}(...\hchoice{q} I_sP_s^{m'})...)$. 
From equation \eqref{eq:decomposition} and theorem \ref{theorem:hineq}, we  derive that, for any $m'>m$,
$$
V_g\hyperc{\pi}{C_{m'}}\geq t_{2m}V_g\hyperc{\pi}{K_m}.
$$

For each $\pi$ and $g$, $V_g\hyperc{\pi}{C}$, being a sum of maximums of continuous functions, is a continuous function over $C$. 
Therefore, the equation above implies \eqref{eq:crowdslower} 

For the proof of the upper bound, theorem \ref{theorem:hineq} gives us
\begin{align*}
V_g\hyperc{\pi}{C_{m'}}\leq&t_{2m}V_g\hyperc{\pi}{K_m}+(1-t_{2m})V_g\hyperc{\pi}{D}.
\end{align*}

Notice that,  $\forall k, j \in \mathbb{N}^*$,  $I_sP_s^k \equiv I_dP_d^k$, and $I_sP_s^{k} \refines I_sP_s^{k+j}=(I_sP_s^{k})P_s^{j}$. 
Thus, lemma \ref{lemma:overpower} yields $V_g \hyperc{\pi}{D}\leq V_g\hyperc{\pi}{I_sP_s^{m+1}}$. Thus 
\begin{align*}
V_g\hyperc{\pi}{C_{m'}}\leq&t_{2m}V_g\hyperc{\pi}{K_m}+(1-t_{2m})V_g\hyperc{\pi}{I_sP_s^{m+1}},
\end{align*}
which, by continuity of $V_g$, implies the upper bound \eqref{eq:crowdsupper}.

Finally, to prove equation \eqref{eq:approx}, it suffices to notice that
\begin{align*}
&(1-t_{2m})V_g\hyperc{\pi}{I_sP_s^{m+1}}\\
\leq&1-t_{2m} &\text{($V_g\hyperc{\pi}{I_sP_s^{m+1}}\leq1$)}\\
=&(1-q)^{m+1}(1-p)^{m} &\text{(by equation \eqref{eq:tneven})}
\end{align*}
\qed
\end{proof}

\section{Further results}
\label{sec:further}
In this section, we explore in more detail the cases in which our operators do not distribute over each other.

\begin{proposition}[Non-distributivity]
\label{prop:nondist}
In general, the following statements do \emph{not} hold
\begin{align*}
 (C_1 \vchoice{p}(C_2 \parallel C_3))&\refines ( (C_1 \vchoice{p} C_2) \parallel (C_1 \vchoice{p} C_3)), \\
 ((C_1 \vchoice{p} C_2) \parallel (C_1 \vchoice{p} C_3) )&\refines ( C_1 \vchoice{p}(C_2 \parallel C_3)), \\
 (C_1 \hchoice{p}(C_2 \parallel C_3))&\refines ( (C_1 \hchoice{p} C_2) \parallel (C_1 \hchoice{p} C_3)), \\
 ((C_1 \hchoice{p} C_2) \parallel (C_1 \hchoice{p} C_3) )&\refines ( C_1 \hchoice{p}(C_2 \parallel C_3)), \\
 (C_1 \hchoice{p}(C_2 \vchoice{q} C_3))&\refines ( (C_1 \hchoice{p} C_2) \vchoice{q} (C_1 \hchoice{p} C_3)), \\
 ((C_1 \hchoice{p} C_2) \vchoice{q} (C_1 \hchoice{p} C_3) )&\refines ( C_1 \hchoice{p}(C_2 \vchoice{q} C_3)).
\end{align*}
\end{proposition}
\begin{proof}
$\bullet$ ($\vchoiceop{p}$ over $\parallel$) Let the following be three compatible channels,
$$
C_1=
\begin{bmatrix}
    1       \\
    1		\\
	1		\\
\end{bmatrix}
, \quad
C_2=
\begin{bmatrix}
    1       & 0   \\
    1       & 0   \\
    0       & 1   \\
\end{bmatrix}
, \quad
C_3=
\begin{bmatrix}
    1       & 0   \\
    0       & 1   \\
    0       & 1   \\
\end{bmatrix}
$$

Then, we have
$$
C_1 \vchoice{\nicefrac{1}{2}} (C_2 \parallel C_3) = 
\begin{bmatrix}
    \nicefrac{1}{2}       \;&\; \nicefrac{1}{2} \;&\; 0 \;&\; 0 \;&\; 0   \\
    \nicefrac{1}{2}       \;&\; 0 \;&\; \nicefrac{1}{2} \;&\; 0 \;&\; 0    \\
    \nicefrac{1}{2}       \;&\; 0 \;&\; 0 \;&\; 0 \;&\; \nicefrac{1}{2}  \\
\end{bmatrix}
$$
$$
(C_1 \vchoice{\nicefrac{1}{2}} C_2) \parallel (C_1 \vchoice{\nicefrac{1}{2}} C_3) = 
\begin{bmatrix}
    \nicefrac{1}{4} \;&\; \nicefrac{1}{4}  \;&\; 0 \;&\; \nicefrac{1}{4} \;&\; \nicefrac{1}{4}  \;&\; 0 \;&\; 0 \;&\; 0 \;&\; 0  \\
	\nicefrac{1}{4} \;&\; 0 \;&\; \nicefrac{1}{4} \;&\; \nicefrac{1}{4} \;&\; 0 \;&\; \nicefrac{1}{4} \;&\; 0\;&\; 0\;&\; 0 \\
	\nicefrac{1}{4} \;&\; 0\;&\; \nicefrac{1}{4} \;&\; 0 \;&\; 0\;&\; 0 \;&\; \nicefrac{1}{4} \;&\; 0 \;&\; \nicefrac{1}{4} \\
\end{bmatrix}
$$

It can be easily checked by solving linear systems that there is no channel $D$ such that
\begin{align*}(C_1 \vchoice{\nicefrac{1}{2}} (C_2 \parallel C_3) ) D &=(C_1 \vchoice{\nicefrac{1}{2}} C_2) \parallel (C_1 \vchoice{\nicefrac{1}{2}} C_3), \mbox{ or}\\
((C_1 \vchoice{\nicefrac{1}{2}} C_2) \parallel (C_1 \vchoice{\nicefrac{1}{2}} C_3)) D &=C_1 \vchoice{\nicefrac{1}{2}} (C_2 \parallel C_3).
\end{align*}

$\bullet$ ($\hchoiceop{p}$ over $\parallel$) Let $\calx=\{x_1,x_2\}$, $\caly=\{y_1,y_2\}$, and let $C_1$,$C_2$ and $C_3$ be of the same type, with input set $\calx$ and
output set $\caly$, given by

$$
C_1=
\begin{bmatrix}
    1       & 0   \\
    0		& 1   \\
\end{bmatrix}
, \quad
C_2=C_3=
\begin{bmatrix}
    0       & 1   \\
    1       & 0   \\
\end{bmatrix}
$$

Then, $C_1 \hchoice{\nicefrac{1}{2}}(C_2 \parallel C_3)$ is a transparent channel, 
and $(C_1 \hchoice{\nicefrac{1}{2}} C_2) \parallel (C_1 \hchoice{\nicefrac{1}{2}} C_3)$ is a null channel. 
Thus, $$(C_1 \hchoice{\nicefrac{1}{2}} C_2) \parallel (C_1 \hchoice{\nicefrac{1}{2}} C_3) \not\refines  C_1 \hchoice{\nicefrac{1}{2}}(C_2 \parallel C_3)$$.

Now let $C_4: \calx \times (\caly \times \caly) \rightarrow \reals$ be given by:
$$C_4 (x_i,(y_j, y_k) )=
\begin{cases}
1 \mbox{, if } i=j=k,\\
0 \mbox{ otherwise}
\end{cases}
$$
for $i,j,k \in \{1,2\}$.

Then, $C_4 \hchoice{\nicefrac{1}{2}}(C_2 \parallel C_3)$ is a null channel, 
and $(C_4 \hchoice{\nicefrac{1}{2}} C_2) \parallel (C_4 \hchoice{\nicefrac{1}{2}} C_3)$ is a transparent channel. 
Thus, $$C_4 \hchoice{\nicefrac{1}{2}}(C_2 \parallel C_3)\not\refines (C_4 \hchoice{\nicefrac{1}{2}} C_2) \parallel (C_4 \hchoice{\nicefrac{1}{2}} C_3) $$.

$\bullet$ ($\hchoiceop{p}$ over $\vchoiceop{p}$) Let $C_1$, $C_2$ and $C_3$ be the same as the proof for $\hchoiceop{p}$ over $\parallel$ above. We have that  $(C_1 \hchoice{\nicefrac{1}{2}} C_2) \vchoice{\nicefrac{1}{2}} (C_1 \hchoice{\nicefrac{1}{2}} C_3)$ is a null channel and $C_1 \hchoice{\nicefrac{1}{2}}(C_2 \vchoice{\nicefrac{1}{2}} C_3)$ is a transparent channel . 
Thus, $$(C_1 \hchoice{\nicefrac{1}{2}} C_2) \vchoice{\nicefrac{1}{2}} (C_1 \hchoice{\nicefrac{1}{2}} C_3) \not\refines  C_1 \hchoice{\nicefrac{1}{2}}(C_2 \vchoice{\nicefrac{1}{2}} C_3)$$.

Now, let, $C_4: \calx \times (\caly \sqcup \caly) \rightarrow \reals$ be given by:
$$C_4 (x_i,(y_j, k) )=
\begin{cases}
\nicefrac{1}{2} \mbox{, if } i=j,\\
0 \mbox{ otherwise}
\end{cases}
$$
for $i,j,k \in \{1,2\}$.

Therefore, $(C_4 \hchoice{\nicefrac{1}{2}} C_2) \vchoice{\nicefrac{1}{2}} (C_1 \hchoice{\nicefrac{1}{2}} C_3)$ is a transparent channel and $C_4 \hchoice{\nicefrac{1}{2}}(C_2 \vchoice{\nicefrac{1}{2}} C_3)$ is a null channel. 
Thus, $$C_4 \hchoice{\nicefrac{1}{2}}(C_2 \vchoice{\nicefrac{1}{2}} C_3)\not\refines (C_4 \hchoice{\nicefrac{1}{2}} C_2) \vchoice{\nicefrac{1}{2}} (C_1 \hchoice{\nicefrac{1}{2}} C_3) $$.
\qed
\end{proof}
}

\end{document}